\documentclass[11pt]{article}%
\usepackage{amsmath}
\usepackage{amsfonts}
\usepackage{amssymb}
\usepackage{graphicx}%
\providecommand{\U}[1]{\protect\rule{.1in}{.1in}}
\newtheorem{theorem}{Theorem}

\newtheorem{example}[theorem]{Example}

\newtheorem{lemma}[theorem]{Lemma}

\newtheorem{proposition}[theorem]{Proposition}
\newtheorem{remark}[theorem]{Remark}

\newenvironment{proof}[1][Proof]{\noindent\textbf{#1.} }{\ \rule{0.5em}{0.5em}}

\begin{document}

\title{What is the Wigner function closest to a given square integrable function?}
\author{J.S. Ben-Benjamin\textbf{\thanks{yonatan@greatwing.com}}
\and N.C. Dias\textbf{\thanks{ncdias@meo.pt}}
\and L. Cohen\textbf{\thanks{leon.cohen@hunter.cuny.edu}}
\and P. Loughlin\textbf{\thanks{loughlin@pitt.edu}}
\and J.N. Prata\textbf{\thanks{joao.prata@mail.telepac.pt}}}
\maketitle

\begin{abstract}
We consider an arbitrary square integrable function $F$ on the
phase space and look for the Wigner function closest to it with
respect to the $L^2$ norm. It is well known that the minimizing
solution is the Wigner function of any eigenvector associated with
the largest eigenvalue of the Hilbert-Schmidt operator with Weyl
symbol $F$. We solve the particular case of radial functions on
the two-dimensional phase space exactly. For more general cases,
one has to solve an infinite dimensional eigenvalue problem. To
avoid this difficulty, we consider a finite dimensional
approximation and estimate the errors for the eigenvalues and
eigenvectors. As an application, we address the so-called Wigner
approximation suggested by some of us for the propagation of a
pulse in a general dispersive medium. We prove that this
approximation never leads to a {\it bona fide} Wigner function.
This is our prime motivation for our optimization problem. As a
by-product of our results we are able to estimate the eigenvalues
and Schatten norms of certain Schatten-class operators. The
techniques presented here may be potentially interesting for
estimating eigenvalues of localization operators in time-frequency
analysis and quantum mechanics.
\end{abstract}


\section{Introduction}

Given two functions $\psi, \phi \in L^2 (\mathbb{R}^d)$, the cross-Wigner function $W(\psi,\phi)$ is given by \cite{Cohen1,Dias1,Lions,Wigner}:
\begin{equation}
W(\psi,\phi) (x,k)  =\frac{1}{(2 \pi)^d} \int_{\mathbb{R}^d} \psi (x + y/2) \overline{\phi (x-y/2)} e^{- i y \cdot k} dy.
\label{eqHS5}
\end{equation}
Here $z=(x,k) \in \mathbb{R}^{2d}$ is interpreted as a phase-space (time-frequency or position-wave number) variable. If $\psi=\phi$ we shall simply write (with some abuse of notation) $W \psi$, meaning $W(\psi,\psi)$ \cite{Wigner}:
\begin{equation}
W \psi (x,k):=W (\psi,\psi) (x,k)= \frac{1}{(2 \pi)^d} \int_{\mathbb{R}^d} \psi \left(x+y/2 \right) \overline{\psi \left(x-y/2 \right)} e^{-i y \cdot k} dy,
\label{eqIntroduction1}
\end{equation}
The Wigner distribution $W \psi$ for a signal $\psi \in L^2 (\mathbb{R}^d)$ is interpreted as a joint phase space representation of the signal.

In the present work, we intend to develop a systematic method to solve the following problem:

 \begin{itemize}
 \item Given some measurable function $F: \mathbb{R}^{2d} \to \mathbb{R}$, which is not a Wigner function, what is the Wigner function $W \psi_0$ closest to it with respect to the $L^2$ norm? In other words, we want to determine $\psi_0 \in L^2 (\mathbb{R}^d)$, such that:
\begin{equation}
\|F - W \psi_0 \|_{L^2 (\mathbb{R}^{2d})}= \inf_{\psi \in L^2 (\mathbb{R}^d)} \|F - W \psi \|_{L^2 (\mathbb{R}^{2d})}.
\label{eqIntroduction0}
\end{equation}
\end{itemize}

This problem and the methods we present may be useful in various contexts. But let us briefly explain our particular motivation for addressing it. In \cite{Cohen2,Loughlin,Loughlin1,Loughlin2} some of us considered the evolution of the Wigner function of a pulse $\psi$ (in $d=1$) given by:
\begin{equation}
W \psi (x,k,t)= \frac{1}{2 \pi} \int_{\mathbb{R}} \psi \left(x+y/2 ,t \right) \overline{\psi \left(x-y/2 ,t \right)} e^{-i y \cdot k} dy,
\label{eqIntroduction2}
\end{equation}
where
\begin{equation}
\psi (x,t)=\int_{\mathbb{R}} G(x-x^{\prime},t) \psi_0(x^{\prime}) dx^{\prime},
\label{eqIntroduction3}
\end{equation}
$\psi_0(x)= \psi(x,0)$ is the pulse at $t=0$, and
\begin{equation}
G(x,t)=\frac{1}{2 \pi}\int_{\mathbb{R}} e^{ikx-i \omega (k) t} dk
\label{eqIntroduction4}
\end{equation}
is the Green's function. In the previous formula $\omega(k)$ is the dispersion relation
\begin{equation}
\omega (k)= \omega_R (k)+i \omega_I (k),
\label{eqIntroduction5}
\end{equation}
which connects the wave number $k$ and the frequency $\omega$. One should understand eq.(\ref{eqIntroduction3}) in the distributional sense $G \in \mathcal{S}^{\prime} (\mathbb{R} \times \mathbb{R})$ and $\psi \in \mathcal{S} (\mathbb{R})$. In \cite{Loughlin1,Loughlin2} the following approximation - called {\it Wigner approximation} - was derived:
\begin{equation}
W \psi (x,k,t) \sim e^{ 2 t \omega_I(k)} W \psi_0 \left(x- \nu(k) t,k \right),
\label{eqIntroduction6}
\end{equation}
where
\begin{equation}
\nu(k)= \omega_R^{\prime} (k)
\label{eqIntroduction7}
\end{equation}
is the group velocity.

The advantage of considering (\ref{eqIntroduction6}) instead of the exact (\ref{eqIntroduction2}) is obvious. In (\ref{eqIntroduction6}), we have a local, computable expression, which has a simple interpretation. Each mode $k$, evolves along a "classical" trajectory with velocity given by the group velocity.

However, with this approximation, one faces a difficulty. As we shall prove in section 6, the expression on the right-hand side of eq.(\ref{eqIntroduction6}) is never the Wigner function $W \phi$ of a signal $\phi \in L^2$ for $t>0$. In this case, we say that that expression is not {\it representable}. However, it may still be a good approximation.

Non-representable functions may also appear, when one conducts "time-varying filtering" \cite{Boudreaux,Jeong} by multiplying a Wigner distribution $W \psi$ by a weighting function of time and frequency:
\begin{equation}
W_{\Gamma} \psi (x,k) = W \psi (x,k) \Gamma (x,k)
\label{eqJeong1}
\end{equation}
The weighting function $\Gamma$ is chosen so that $W_{\Gamma} \psi$ has some optimal time-frequency concentration.

The Wigner transform is a fundamental instrument in the spectral esti-
mation of non-stationary signals. In some situations a non-representable
phase space function may appear, for instance: in multitaper estimation
\cite{Bayram}, specially when combined with reassignment \cite{Xiao} and in the Wigner
distribution of linear signal spaces \cite{Hlawatsch1,Hlawatsch2}.

Motivated by these three situations, we intend to study the problem stated above. If a given real-valued function $F \in L^2 (\mathbb{R}^{2d})$ is not representable, that is if there is no $\psi \in L^2 (\mathbb{R}^d)$ such that $F=W \psi$, then what is the Wigner function $W \psi_0$ "closest" to $F$? Since, via Moyal's identity \cite{Moyal}, Wigner functions belong to $L^2 (\mathbb{R}^{2d})$, it seems natural to require proximity in the $L^2$-norm. This least squares problem has emerged in other contexts such as Bessel multipliers \cite{Balazs1,Balazs2}, and time-varying filtering and signal estimation using Wigner functions \cite{Boudreaux,Jeong} and short-time Fourier transforms \cite{Griffin}.

In a companion paper \cite{Ben} we proved that such a minimizer always exists, although it may not be unique. Moreover, we give an explicit construction of the minimizers. Nevertheless, it may be difficult to obtain it. This is because the construction requires the computation of the spectrum and the eigenspace associated with the largest eigenvalue of the self-adjoint Hilbert-Schmidt operator $\widehat{F}$ with Weyl symbol $F$. Since, in general, the spectrum of $\widehat{F}$ may be infinite, albeit countable, this may prove to be a difficult task. If such is the case, we choose to replace the infinite dimensional eigenvalue problem by a finite dimensional one. We then give precise estimates for the errors of the eigenvalues and eigenvectors of the truncated problem. As a by-product of these estimates we can approximate the eigenvalues and Schatten norms of certain Schatten-class operators \cite{Exner}.

These techniques may be potentially interesting in other contexts. For instance, in quantum mechanics, mixed states are represented by positive trace-class operators - the so-called density matrices. In general, it is very difficult to assess whether a given operator acting on an infinite dimensional Hilbert space is positive. The techniques developed here allow us to iteratively compute a sequence of positive trace-class operators which approximate the given operator. Also, as we will point out, this optimization problem is intimately related to localization (Toeplitz) operators \cite{Bracken,Daubechies,Lieb,Ramanathan}.

Here is a brief summary of the paper. In the next section, we introduce the main concepts related to the spectrum and Weyl transform of Hilbert-Schmidt operators. In section 3, we present the solution for the optimization problem, and we solve exactly a particular case in $d=1$ in section 4 (this is roughly speaking the case of "radial" functions). In section 5, we present the main results of this work. We consider the truncated eigenvalue problem and derive precise estimates for the errors of the eigenvalues and eigenvectors. In section 6, we go back to the Wigner approximation. We show that the Wigner approximation is never representable. We illustrate our results with a simple example. In section 7, we address the problem of obtaining approximately the spectrum and the Schatten norm for some Schatten-class operators. Finally, in section 8, we present our conclusions and discuss the possibility of applying our results to quasi-distributions other than the Wigner distribution.

\section*{Notation}
The complex conjugate of a number $c$ is written $\overline{c}$. If $\widehat{A}$ is a linear operator acting on some Hilbert space, then we denote by $Ker(\widehat{A})$ its kernel. The inner product and the norm on $L^2 (\mathbb{R}^d)$ are
\begin{equation}
< \psi, \phi >_{L^2 (\mathbb{R}^d)} = \int_{\mathbb{R}^d} \psi (x) \overline{\phi(x)} dx
\label{eqinnerproduct}
\end{equation}
and
~\begin{equation}
|| \psi||_{L^2 (\mathbb{R}^d)} =\left( \int_{\mathbb{R}^d} |\psi (x)|^2 dx\right)^{1/2},
 \label{eqnorm}
\end{equation}
respectively. We denote by $\mathcal{S} (\mathbb{R}^d)$ the Schwartz class of test functions and by $\mathcal{S}^{\prime} (\mathbb{R}^d)$ its dual - the tempered distributions. We shall denote by $|| \cdot||_{l^2}$ the norm for the spaces of square-summable sequences
\begin{equation}
l^2 (\mathbb{N})= \left\{ c= \left\{c_n \right\}_n: ~||c||_{l^2}^2=\sum_{n=1}^{\infty} |c_n|^2 < \infty \right\}
\label{eqsequences1}
\end{equation}
and
\begin{equation}
l^2 (\mathbb{N}^2)= \left\{ \mathbb{F}= \left\{f_{n,m} \right\}_{n,m}: ~||\mathbb{F}||_{l^2}^2=\sum_{n,m=1}^{\infty} |f_{n,m}|^2 < \infty \right\}
\label{eqsequences2}
\end{equation}

The Fourier-Plancherel transform of $f \in L^1 (\mathbb{R}^d) \cap L^2 (\mathbb{R}^d)$ is defined by:
\begin{equation}
(\mathcal{F}f) (k):= \frac{1}{(2 \pi)^{d/2}} \int_{\mathbb{R}^d} f(x) e^{-i k \cdot x} dx.
\label{eqFourierTransform}
\end{equation}

\section{Hilbert-Schmidt operators and Weyl transform}

In this section we review some well known definitions and results of Hilbert-Schmidt operators and the Weyl transform. For more details the reader should refer to \cite{Exner,Wong}.

\subsection{Hilbert-Schmidt operators}

A Hilbert-Schmidt operator $\widehat{A}: \mathcal{H} \to \mathcal{H}$ on a separable Hilbert space $\mathcal{H}$ is a linear operator such that \cite{Exner}
\begin{equation}
\sum_i || \widehat{A}e_i||_{\mathcal{H}}^2 < + \infty,
\label{eqHSOp1}
\end{equation}
for any orthonormal basis $\left\{e_i \right\}_i$.

We denote by $S_2 (\mathcal{H})$ the set of Hilbert-Schmidt operators. This is a Hilbert space with inner product
\begin{equation}
< \widehat{A},\widehat{B}>_{S_2 (\mathcal{H})}:= \sum_i < \widehat{A}e_i,\widehat{B}e_i>_{\mathcal{H}}
\label{eqHSOp2}
\end{equation}
and norm:
\begin{equation}
|| \widehat{A}||_{S_2 (\mathcal{H})}^2:= \sum_i || \widehat{A}e_i||_{\mathcal{H}}^2.
\label{eqHSOp3}
\end{equation}
It can be shown that the previous expressions do not depend on the orthonormal basis.

Hilbert-Schmidt operators are compact operators \cite{Exner,Reed}. If $\widehat{A} \in S_2 (\mathcal{H})$ is self-adjoint, then it admits the spectral decomposition:
\begin{equation}
\widehat{A} = \widehat{A}_+ + \widehat{A}_-,
\label{eqHSOp4}
\end{equation}
where
\begin{equation}
\widehat{A}_{\pm} := \sum_{j \in \mathbb{T}_{\pm}} \lambda_j \widehat{P}_j
\label{eqHSOp5}
\end{equation}
are the positive $(+)$ and the negative $(-)$ parts, $\mathbb{T}_{+}$ and $\mathbb{T}_{-}$ are (possibly finite) sets of integer indices labelling the positive and the negative eigenvalues, respectively,  $\left\{\lambda_j \right\}_{j \in \mathbb{T}_+}$ are the positive eigenvalues, written as a decreasing sequence
\begin{equation}
\lambda_1 > \lambda_2 > \lambda_3 > \cdots >0,
\label{eqHSOp6}
\end{equation}
$\left\{\lambda_j \right\}_{j \in \mathbb{T}_-}$ are the negative eigenvalues, written as an increasing sequence
\begin{equation}
\lambda_{-1} < \lambda_{-2} < \lambda_{-3} < \cdots <0,
\label{eqHSOp7}
\end{equation}
and $\widehat{P}_j : \mathcal{H} \to \mathcal{H}_j$ is the orthogonal projection onto the eigenspace $\mathcal{H}_j$ associated with the eigenvalue $\lambda_j$. Each eigenspace is finite dimensional: $n_j= dim (\mathcal{H}_j) < + \infty$.

The Hilbert space splits into the Hilbert sum:
\begin{equation}
\mathcal{H} = Ker (\widehat{A}) \oplus \left(\mathcal{H}_{-1}\oplus \mathcal{H}_{-2} \oplus \cdots \right) \oplus \left(\mathcal{H}_1\oplus \mathcal{H}_2 \oplus \cdots \right)
\label{eqHSOp8}
\end{equation}
By choosing orthonormal basis in each eigenspace $\mathcal{H}_j$, we can rewrite (\ref{eqHSOp5}) as
\begin{equation}
\widehat{A}_{\pm} := \sum_{j \in \mathbb{U}_{\pm}} \mu_{\pm j} \widehat{P}_{\pm j}
\label{eqHSOp9}
\end{equation}
where $\widehat{P}_{\pm j}$ is the projector in the direction of the vector $e_{\pm j}$ of the orthonormal set of eigenvectors, with $<e_{\alpha},e_{\beta}>_{\mathcal{H}}= \delta_{\alpha, \beta}$ for all $\alpha, \beta \in \mathbb{U}= \mathbb{U}_+ \cup \mathbb{U}_-$, and $\widehat{A} e_{\alpha}= \mu_{\alpha} e_{\alpha}$. The eigenvalues $\left\{\mu_{\alpha}\right\}_{\alpha \in \mathbb{U}}$ are the same as $\left\{\lambda_{\alpha} \right\}_{\alpha \in \mathbb{T}}$, but they are not all necessarily distinct. This happens whenever some eigenvalue is degenerate $(n_j >1)$.

\subsection{The Weyl transform}

In this work we deal with the case $\mathcal{H} = L^2 (\mathbb{R}^d)$. A Hilbert-Schmidt operator $\widehat{A} \in S_2 \left(L^2(\mathbb{R}^d) \right)$ is given by
\begin{equation}
(\widehat{A} \psi )(x) := \int_{\mathbb{R}^d} K_A (x,y) \psi (y) dy, \hspace{1 cm} \psi \in L^2(\mathbb{R}^d)
\label{eqHS1}
\end{equation}
with a kernel $K_A \in L^2(\mathbb{R}^d \times \mathbb{R}^d)$.

The Weyl transform \cite{Cohen1,Wong} is a linear map
\begin{equation}
\mathcal{W}: S_2 \left(L^2(\mathbb{R}^d) \right) \to L^2(\mathbb{R}^{2d})
\label{eqWigner1}
\end{equation}
defined by
\begin{equation}
\mathcal{W}(\widehat{A} )(x,k):= \int_{\mathbb{R}^d} K_A(x+y/2,x-y/2) e^{-i y \cdot k} dy.
\label{eqHS2}
\end{equation}
The function $\mathcal{W}(\widehat{A} )(x,k)$ is called the Weyl symbol of $\widehat{A}$. $\mathcal{W}$ is a bijection with inverse $\mathcal{W}^{-1}$. Thus, given a symbol $A(x,k) \in L^2 (\mathbb{R}^{2d})$, the associated Weyl operator is the Hilbert-Schmidt operator
\begin{equation}
\left((\mathcal{W}^{-1} A) \psi\right) (x)=  \frac{1}{(2 \pi)^d} \int_{\mathbb{R}^d} \int_{\mathbb{R}^d}  A \left(\frac{x+y}{2},k \right) e^{ik \cdot(x-y)} \psi (y) dk dy,
\label{eqHS3}
\end{equation}
defined for all $ \psi \in L^2(\mathbb{R}^d)$. A Hilbert-Schmidt operator $\widehat{A}$ is self-adjoint if and only if its symbol $A(x,k)$ is a real function. In this case, the operator $\widehat{A}$ has the spectral decomposition (\ref{eqHSOp4},\ref{eqHSOp5},\ref{eqHSOp9}) with a positive $(\widehat{A}_+)$ and a negative part $(\widehat{A}_-)$. Using the Weyl transform, we may thus define:
\begin{equation}
A= \mathcal{W} (\widehat{A}) = A_+ + A_-,
\label{eqHS3.1}
\end{equation}
where
\begin{equation}
A_{\pm}= \mathcal{W} (\widehat{A}_{\pm}) .
\label{eqHS3.2}
\end{equation}
An important case is when the operator $\widehat{A}_{\psi,\phi}$ (\ref{eqHS1}) is the rank one operator with kernel
\begin{equation}
K_{\psi,\phi} (x,y) = (\psi \otimes \overline{\phi})(x,y) =  \psi (x)\overline{\phi (y)},
\label{eqHS4}
\end{equation}
with $\psi , \phi \in L^2 (\mathbb{R}^d)$. The associated Weyl symbol is (up to a multiplicative constant) the non-diagonal Wigner function (\ref{eqHS5}):
\begin{equation}
\mathcal{W}(\widehat{A}_{\psi,\phi})(x,k)= (2 \pi)^d W(\psi,\phi) (x,k)= \int_{\mathbb{R}^d} \psi (x + y/2) \overline{\phi (x-y/2)} e^{- i y \cdot k} dy.
\label{eqextra1}
\end{equation}

Weyl operators and Wigner functions are also related via the following remarkable formula. Let $\widehat{F}$ be some Hilbert-Schmidt operator with Weyl symbol $F = \mathcal{W} (\widehat{F})$, and let $\psi, \phi \in L^2(\mathbb{R}^d)$. Then we have \cite{Gosson}:
\begin{equation}
\langle \widehat{F} \psi, \phi \rangle_{L^2 (\mathbb{R}^d)} = \langle F, W(\phi,\psi) \rangle_{L^2 (\mathbb{R}^{2d})}
\label{eqextra2}
\end{equation}

In the previous identity, let us choose $\widehat{F}= \widehat{A}_{\psi_1,\psi_2}$ as in (\ref{eqHS4},\ref{eqextra1}):
\begin{equation}
\left(\widehat{A}_{\psi_1,\psi_2} \eta \right) (x) = \langle \eta ,\psi_2 \rangle_{L^2 (\mathbb{R}^d)} \psi_1 (x),
\label{eqMoyalA}
\end{equation}
for all $\eta \in L^2 (\mathbb{R}^d)$. From (\ref{eqextra2},\ref{eqMoyalA}), we have:
\begin{equation}
\begin{array}{c}
\langle \widehat{A}_{\psi_1,\psi_2} \phi_2, \phi_1 \rangle_{L^2 (\mathbb{R}^d)} = \langle A_{\psi_1,\psi_2} , W(\phi_1,\phi_2) \rangle_{L^2 (\mathbb{R}^{2d})}\\
\\
\Leftrightarrow \langle \phi_2,\psi_2 \rangle_{L^2 (\mathbb{R}^d)} \langle \psi_1,\phi_1 \rangle_{L^2 (\mathbb{R}^d)} = (2 \pi)^d \langle W(\psi_1,\psi_2), W(\phi_1,\phi_2) \rangle_{L^2 (\mathbb{R}^{2d})} ,
\end{array}
\label{eqMoyalB}
\end{equation}
and we obtain Moyal's identity \cite{Moyal}:
\begin{equation}
\langle W(\psi_1,\psi_2),W(\phi_1,\phi_2) \rangle_{L^2 (\mathbb{R}^{2d})}=\frac{1}{(2 \pi)^d} \langle \psi_1,\phi_1 \rangle_{L^2 (\mathbb{R}^d)} \langle \phi_2,\psi_2 \rangle_{L^2 (\mathbb{R}^d)}.
\label{eqMoyalIdentity1}
\end{equation}
As a consequence of this, we have:

\begin{lemma}\label{LemmaOrthormalBasis}
Let $\left\{e_n \right\}_n$ be an orthonormal basis of $L^2 (\mathbb{R}^d)$. Then the functions $\left\{(2 \pi)^{d/2} W(e_n,e_m) \right\}_{n,m}$ form an orthonormal basis of $L^2 (\mathbb{R}^{2d})$.
\end{lemma}

\begin{proof}
From Moyal's identity, we have:
\begin{equation}
\begin{array}{c}
\langle W(e_n, e_m), W (e_k,e_l) \rangle_{L^2 (\mathbb{R}^{2d})}= \frac{1}{(2 \pi)^d} \langle e_n, e_k \rangle_{L^2 (\mathbb{R}^d)} \langle e_l, e_m \rangle_{L^2 (\mathbb{R}^d)}=\\
\\
= \frac{1}{(2 \pi)^d} \delta_{n,k} \delta_{l,m},
\end{array}
\label{eqOrthormalBasis1}
\end{equation}
which shows that $\left\{(2 \pi)^{d/2} W(e_n,e_m) \right\}_{n,m}$ are an orthonormal set.

Next, assume that $F \in L^2 (\mathbb{R}^{2d})$ is such that
\begin{equation}
\langle F,  W(e_n, e_m) \rangle_{L^2 (\mathbb{R}^{2d})}=0,
\label{eqOrthormalBasis2}
\end{equation}
for all $n,m$. Let $\widehat{F}= \mathcal{W}^{-1} (F)$ be the Hilbert-Schmidt operator with Weyl symbol $F$. From (\ref{eqextra2}), we have for all $n,m$:
\begin{equation}
0= \langle \widehat{F} e_n, e_m \rangle_{L^2 (\mathbb{R}^d)}
\label{eqOrthormalBasis3}
\end{equation}
Since $\left\{e_n \right\}_n$ is an orthonormal basis of $L^2 (\mathbb{R}^d)$, this is possible if and only if $ \widehat{F} =0$ and $F \equiv 0$. Consequently the orthonormal set $\left\{(2 \pi)^{d/2} W(e_n,e_m) \right\}_{n,m}$ is complete.
\end{proof}

\section{The optimization problem}

Let $F \in L^2 (\mathbb{R}^{2d})$ and $\mathcal{E}=\left\{e_n \right\}_{n \in \mathbb{N}}$ be some orthonormal basis of $L^2 (\mathbb{R}^d)$. We may thus write:
\begin{equation}
F(z)= \sum_{n,m \in \mathbb{N}} f_{n,m} W (e_n, e_m) (z),
\label{eq1}
\end{equation}
where the coefficients $f_{n,m}$ are given by:
\begin{equation}
f_{n,m}=(2 \pi)^d <F,W(e_n,e_m)>_{L^2 (\mathbb{R}^{2d})},
\label{eq13}
\end{equation}
If $F$ is a real function, then
\begin{equation}
\overline{f_{n,m}}=f_{m,n}, \hspace{1 cm} \forall n, m \in \mathbb{N}.
\label{eq2}
\end{equation}
Moreover, we have (cf. (\ref{eqMoyalIdentity1})):
\begin{equation}
||F||_{L^2(\mathbb{R}^{2d})}^2 = \frac{1}{(2 \pi)^d}\sum_{n,m \in \mathbb{N}} |f_{n,m}|^2 < \infty.
\label{eq3}
\end{equation}
Given some $\psi \in L^2 (\mathbb{R}^d)$, we can also expand it in the basis $\mathcal{E}$:
\begin{equation}
\psi (x) = \sum_{n \in \mathbb{N}} c_n e_n (x),
\label{eq4}
\end{equation}
with
\begin{equation}
||\psi||_{L^2(\mathbb{R}^d)}^2 = \sum_{n \in \mathbb{N}} |c_n|^2 < \infty.
\label{eq5}
\end{equation}
This entails:
\begin{equation}
W \psi (z)= \sum_{n, m \in \mathbb{N}} c_n \overline{c_m} W (e_n, e_m)(z).
\label{eq6}
\end{equation}

In the sequel, $\mathbb{F}$ denotes the infinite matrix with coefficients $\left\{f_{n,m} \right\}_{n, m \in \mathbb{N}}$ and $c$ is the column vector $c^T=(c_1,c_2, c_3, \cdots)$. In view of (\ref{eq3},\ref{eq5}), we have that $\mathbb{F} \in l^2(\mathbb{N}^2)$ and $c \in l^2(\mathbb{N})$. We may regard $\mathbb{F}$ as a bounded linear operator $ l^2(\mathbb{N}) \to l^2(\mathbb{N})$, $c \mapsto \mathbb{F}c$ with operator norm
\begin{equation}
||\mathbb{F}||_{Op}:= sup_{c \in l^2(\mathbb{N}) \backslash \left\{0 \right\}} \frac{||\mathbb{F}c||_{l^2}}{||c||_{l^2}},
\label{eq6.1}
\end{equation}
with $|| \cdot||_{l^2}$ the $l^2$ norm. The boundedness is easily established by the fact that the operator norm is dominated by the $l^2$-norm:
\begin{equation}
||\mathbb{F}||_{Op} \le ||\mathbb{F}||_{l^2} =(2\pi)^{d/2} ||F||_{L^2 (\mathbb{R}^{2d})}.
\label{eq6.2}
\end{equation}

\begin{remark}\label{RemarkRepresentations}
It is interesting to remark that we have three representations of the same object. First of all, we have a self-adjoint Hilbert-Schmidt operator $\widehat{F}=\widehat{F}_+ +\widehat{F}_-$ acting on the Hilbert space $L^2 (\mathbb{R}^d)$. From the Weyl transform, we obtain its counterpart in phase space $\mathcal{W} (\widehat{F})= F =F_+ + F_-$. Finally, using the expansion (\ref{eq1},\ref{eq13}) of $F$ in some orthonormal basis of Wigner functions we obtain yet another representation of $\widehat{F}$ - the matrix $\mathbb{F} \in l^2 (\mathbb{N}^2)$. The important thing is that they all have the same spectrum. That is the eigenvalues of the operator $ \widehat{F}$ are the same as those of the matrix $\mathbb{F}$. Moreover, they also coincide with the eigenvalues of the function $F$ regarded as a pseudodifferential operator $F_{\ast} : L^2 (\mathbb{R}^{2d}) \to L^2 (\mathbb{R}^{2d})$, which acts on $G \in L^2 (\mathbb{R}^{2d})$ as $F_{\ast}(G)= F \ast G := \mathcal{W} \left( \mathcal{W}^{-1} (F) \cdot \mathcal{W}^{-1} (G)\right)$, where $\ast$ is the Moyal star product.
\end{remark}

Recall that we want to find the Wigner function $W \psi^{(0)}$ closest to $F$ in $L^2(\mathbb{R}^{2d})$. This amounts to minimizing the following functional:

\begin{equation}
\mathcal{L} ( c ) := ||F- W \psi||_{L^2(\mathbb{R}^{2d})}^2
\label{eq7}
\end{equation}
From (\ref{eq1},\ref{eq6}) and Moyal's identity, we obtain:
\begin{equation}
\mathcal{L} ( c )
=||F||_{L^2(\mathbb{R}^{2d})}^2  - \frac{2}{(2 \pi)^d} \sum_{n,m \in \mathbb{N}} \overline{c_n} f_{n,m} c_m + \frac{1}{(2 \pi)^d} \left(\sum_{n \in \mathbb{N}} |c_n|^2 \right)^2.
\label{eq8}
\end{equation}

The following can be found in \cite{Jeong} for non-degenerate spectrum. We consider it here for completeness.

\begin{theorem}\label{TheoremMinimizer1}
If $F_+  \equiv 0$, then the minimizing function of (\ref{eq7}) is $W \psi^{(0)} \equiv 0$. Otherwise, it is given by (\ref{eq6}), where $c^{(0)}$ is any eigenvector of $\mathbb{F}$ associated with the largest eigenvalue $\lambda_{max} $:
\begin{equation}
\mathbb{F} c^{(0)}= \lambda_{max} c^{(0)}.
\label{eq8}
\end{equation}
Moreover, the normalization of $\psi^{(0)}$ is such that
\begin{equation}
||\psi^{(0)}||_{L^2 (\mathbb{R}^d)}^2= ||c^{(0)}||_{l^2 }^2 = \lambda_{max}.
\label{eq9}
\end{equation}
The minimal distance is then:
\begin{equation}
\begin{array}{c}
\mbox{min}_{\psi \in L^2 (\mathbb{R}^d) } ||F- W \psi||_{L^2(\mathbb{R}^{2d})}^2= \mathcal{L} ( c^{(0)} ) = \\
\\
=||F||_{L^2 (\mathbb{R}^{2d})}^2 - \frac{\lambda_{max}^2}{(2 \pi)^d}=  ||F||_{L^2 (\mathbb{R}^{2d})}^2 - \frac{||\psi^{(0)}||_{L^2 (\mathbb{R}^d)}^4}{(2 \pi)^d}
\end{array}
\label{eq9.1}
\end{equation}

\end{theorem}

\begin{proof}
In \cite{Ben} we proved the existence of a global minimizer. In the calculus of variations, if $c^{(0)}$ is a minimizer, then the functional $\mathcal{L}$ has to be stationary at $c^{(0)}$ \cite{Jahn,Jost}. In other words, the Fr\'echet derivative of (\ref{eq7}) with respect to the real and imaginary parts of $c_n $, or equivalently with respect to $c_n$ and $\overline{c_n}$, have to vanish identically. Imposing a vanishing derivative with respect to $c_n$ is equivalent to doing the same with respect to $\overline{c_n}$: one equation is obtained from the other by complex conjugation. In particular, the Fr\'echet derivative of (\ref{eq7}) with respect to $\overline{c_n}$ yields:
\begin{equation}
\frac{\partial {\mathcal L}}{\partial \overline{c_n}}= - \frac{2}{(2 \pi)^d} \sum_{m \in \mathbb{N}} f_{n,m} c_m + \frac{2}{(2 \pi)^d} ||c||_{l^2 }^2 c_n.
\label{eq10}
\end{equation}
Thus the stationarity condition at $c^{(0)}$ becomes:
\begin{equation}
\sum_{m \in \mathbb{N}} f_{n,m} c_m^{(0)}=||c^{(0)}||_{l^2 }^2 c_n^{(0)},
\label{eq11}
\end{equation}
for all $n \in \mathbb{N}$. In other words, $c^{(0)}$ is an eigenvector of $\mathbb{F}$ with eigenvalue $||c^{(0)}||_{l^2 }^2 \ge 0$. If $F_+ \equiv 0$, then $||c^{(0)}||_{l^2 }^2=0$ and the minimizing solution is $W \psi^{(0)} \equiv 0$. Alternatively, if $F_+$ is not identically zero, then it follows that
\begin{equation}
\mathcal{L} (c^{(0)})= ||F||_{L^2(\mathbb{R}^{2d})}^2 - \frac{||c^{(0)}||_{l^2 }^4}{(2 \pi)^d},
\label{eq12}
\end{equation}
where $||c^{(0)}||_{l^2 }^2$ is one of the eigenvalues of $\mathbb{F}$. Obviously, (\ref{eq12}) is minimal if $||c^{(0)}||_{l^2 }^2$ is equal to the largest eigenvalue $\lambda_{max}$.
\end{proof}

The problem we want to address is how to compute $\lambda_{max}$ and $c^{(0)}$. If the matrix $\mathbb{F}$ is finite dimensional, there are good approximation techniques to compute eigenvalues and eigenvectors such as the Rayleigh-Ritz method \cite{Jia}. Here however we want to focus on the infinite dimensional case.

Several approaches can be considered. In \cite{Boudreaux,Jeong} the authors considered the "weighted Wigner distribution" $W_{\Gamma} \psi (x,k)$ (\ref{eqJeong1}). To deal with the infinite dimensional problem, they chose to use a discrete-time Wigner function.

Our strategy is different and consists of truncating the infinite eigenvalue problem at some finite order $N \in \mathbb{N}$. We cannot hope to obtain (in general) the exact solution, but we can derive precise estimates for the truncated version.

Let us then explain our approach in more detail. We are given some real-valued function $F \in L^2 (\mathbb{R}^{2d})$. We evaluate its norm $||F||_{L^2 (\mathbb{R}^{2d})}$, choose a particular orthonormal basis $\mathcal{E}$ and compute the expansion coefficients (\ref{eq13}) for $n,m=1,2, \cdots, N$.

We next obtain the truncated $N \times N$ matrix $\mathbb{F}^{(N)}$. To fix the order $N$, we use the following criterion. Let
\begin{equation}
F^{(N)} (z)= \sum_{n,m=1}^N f_{n,m} W(e_n,e_m) (z)
\label{eq14}
\end{equation}
be the truncated function. The relative error is given by:
\begin{equation}
0 < \frac{||F-F^{(N)} ||_{L^2 (\mathbb{R}^{2d})}}{||F||_{L^2 (\mathbb{R}^{2d})}} < 1
\label{eq15}
\end{equation}
We choose $N = N(\epsilon)$ to be the smallest order for which the relative error is smaller than some given value $\epsilon$:
\begin{equation}
\frac{||F-F^{(N)} ||_{L^2 (\mathbb{R}^{2d})}}{||F||_{L^2 (\mathbb{R}^{2d})}} < \epsilon.
\label{eq16}
\end{equation}
Next, we compute the eigenvalues and eigenvectors of the truncated matrix $\mathbb{F}^{(N)}$. As before, we write its positive eigenvalues as a decreasing sequence
\begin{equation}
\mu_1^{(N)} \ge \mu_2^{(N)} \ge \cdots \ge \mu_{N_+}^{(N)},
\label{eq17}
\end{equation}
and its negative eigenvalues as an increasing sequence
\begin{equation}
\mu_{-1}^{(N)} \le \mu_{-2}^{(N)} \le \cdots \le \mu_{-N_-}^{(N)}.
\label{eq17.1}
\end{equation}
Notice that zero may also be an eigenvalue of $\mathbb{F}^{(N)}$. Let $N_K$ denote the dimension of the kernel of $\mathbb{F}^{(N)}$. If it is trivial, then $N_K=0$. We thus have:
\begin{equation}
N_+ + N_- +N_K =N.
\label{eq17.2.0}
\end{equation}

We denote by $\mathcal{C}^{(N)} = \left\{c_j^{(N)} \right\} = \left\{ c_1^{(N)}, \cdots, c_{N_+}^{(N)},c_{-1}^{(N)}, \cdots, c_{-N_-}^{(N)}\right\}$ an orthonormal set of eigenvectors:
\begin{equation}
\mathbb{F}^{(N)} c_{j}^{(N)} = \mu_{j}^{(N)} c_{j}^{(N)}, \hspace{1 cm} \overline{c_j^{(N)}} \cdot c_k^{(N)}= \delta_{j,k} ,
\label{eq17}
\end{equation}
with $j,k \in \left\{-N_-, \cdots,-1,1, \cdots, N_+ \right\} $. We also consider an orthonormal set $\mathcal{D}^{(N)}= \left\{d_j^{(N)} \right\}= \left\{ d_1^{(N)}, \cdots, d_{N_K}^{(N)}\right\}$ spanning $Ker(\mathbb{F}^{(N)})$. Thus $\mathcal{B}^{(N)}:=\mathcal{C}^{(N)} \cup \mathcal{D}^{(N)}$ is an orthonormal basis for $\mathbb{C}^N$:
\begin{equation}
\mathcal{B}^{(N)}= \left\{ c_1^{(N)}, \cdots, c_{N_+}^{(N)}, d_1^{(N)}, \cdots, d_{N_K}^{(N)},  c_{-N_-}^{(N)}, \cdots, c_{-1}^{(N)} \right\}  .
\label{eq17.1}
\end{equation}
We can write $\mathcal{B}^{(N)}$ in the compact form
\begin{equation}
\mathcal{B}^{(N)}= \left\{e_{\alpha}^{(N)} \right\}_{\alpha \in \mathcal{A}^{(N)}}
\label{eq17.1.A}
\end{equation}
where $\mathcal{A}^{(N)}= \left\{1, \cdots, N \right\}$, and where
\begin{equation}
e_{\alpha}^{(N)} = \left\{
\begin{array}{l l}
c_{\alpha}^{(N)} & , ~\alpha=1, \cdots, N_+ \\
d_{\alpha-N_+}^{(N)} & , ~\alpha=N_+ +1, \cdots, N_+ +N_K \\
c_{\alpha -N-1}^{(N)} & , ~\alpha=N_+ + N_K + 1, \cdots, N
\end{array}
\right.
\label{eq17.1.B}
\end{equation}
We thus have
\begin{equation}
 \overline{e_{\alpha}^{(N)}} \cdot e_{\beta}^{(N)} = \delta_{\alpha , \beta},
\label{eq17.2}
\end{equation}
for all $\alpha, \beta \in \mathcal{A}^{(N)}$.

Likewise, we consider an orthonormal basis for $l^2 (\mathbb{N})$
\begin{equation}
\mathcal{B}= \left\{ c_1,  c_2, \cdots,d_1, d_2,\cdots,  c_{-2},  c_{-1} \right\} = \left\{e_{\alpha} \right\}_{\alpha \in \mathcal{A}},
\label{eq17.3}
\end{equation}
where
\begin{equation}
\mathbb{F} c_j = \mu_j c_j, \hspace{1 cm} \mathbb{F} c_{-j} = \mu_{-j} c_{-j},
\label{eq17.4}
\end{equation}
for $j=1,2, \cdots$, and where $d_1,d_2, \cdots$ span $Ker (\mathbb{F})$. Moreover, $\mathcal{A}$ is a countable index set that yields all the elements of $\mathcal{B}$ in a unified description.

Our purpose is to approximate the largest eigenvalue $\mu_1$ with $\mu_1^{(N)}$ and, if possible, the eigenvector $c_1$ (which is equal to $c^{(0)}$ in the notation of Theorem \ref{TheoremMinimizer1}) with $c_1^{(N)}$. To have some control over the quality of the approximation, we shall derive estimates for $|\mu_1-\mu_1^{(N)}|$ and $||c_1-c_1^{(N)}||_{l^2}$ in terms of $\epsilon$.

\section{A particular case in $d=1$}

Before we address the general case, we consider a simplified situation in one-dimension. In this case a natural choice of orthogonal basis would be the hermite functions $(h_n)$, the eigenfunctions of the harmonic oscillator \cite{Thangavelu}:
\begin{equation}
h_n(x)= (-1)^n e^{x^2/2} \frac{d^n}{d x^n} e^{- x^2}, \hspace{1 cm} n=0,1,2, \cdots
\label{eqHermiteFunctions1A}
\end{equation}
They can be normalized as follows:
\begin{equation}
e_n(x)= \frac{h_n(x)}{(2^n n!\sqrt{\pi})^{1/2}}, \hspace{1 cm} < e_n,e_m >_{L^2 (\mathbb{R})}= \delta_{n,m}
\label{eqHermiteFunctions1B}
\end{equation}

The corresponding Wigner functions are \cite{Wong}:
\begin{equation}
W(e_{j+k},e_j)(z)= \frac{(-1)^j}{\pi} \sqrt{\frac{j!}{(j+k)!} }\overline{a}^k L_j^k (2 |z|^2)e^{-|z|^2}, \hspace{1 cm} j,k=0,1,2, \cdots,
\label{eqHermiteFunctions1}
\end{equation}
where
\begin{equation}
a= \sqrt{2}(x+ik)
\label{eqHermiteFunctions2}
\end{equation}
and $L_j^k $ are the Laguerre polynomials:
\begin{equation}
\begin{array}{c}
L_n^{\alpha} (x)= \frac{x^{- \alpha} e^x}{n!} \frac{d^n}{d x^n} (e^{-x} x^{\alpha +n})=\\
\\
= \sum_{k=0}^n \frac{(\alpha +n)(\alpha +n-1) \cdots (\alpha +k +1)}{(n-k)! k!} (-x)^k,
\end{array}
\label{eqHermiteFunctions3}
\end{equation}
for $x >0$ and $n=0,1,2, \cdots$. The Wigner functions $W(e_j, e_{j+k})$ can be obtained from (\ref{eqHermiteFunctions1}), by noticing that
\begin{equation}
W(f,g)(z)= \overline{W(g,f)(z)}.
\label{eqHermiteFunctions4}
\end{equation}

We should add a word of caution concerning the notation of eq.(\ref{eqHermiteFunctions2}). In complex analysis the letter $z$ is used to denote complex numbers such the one in (\ref{eqHermiteFunctions2}), $\sqrt{2}(x+ik) \in \mathbb{C}$. This is also the notation in analytic (or poly-analytic) time-frequency representations, such as the Bargmann (or poly-Bargmann) transforms \cite{Abreu1,Abreu4,Haimi}. However, here we have reserved the letter $z$ to denote the {\it real} phase space point $z=(x,k) \in \mathbb{R}^2$. This is the reason for choosing the notation $a$. Notice that this is closer to the physicists notation, where $a$, $\overline{a}$ can be seen as the Weyl symbols of the annihilation and creation operators, respectively.

Before we continue, let us make the following observation. The diagonal Wigner functions associated with the Hermite functions (\ref{eqHermiteFunctions1}) are radial:
\begin{equation}
W e_n (z)= \frac{(-1)^n}{\pi} L_n^0 (2 |z|^2) e^{- |z|^2}.
\label{eqHermiteFunctions1.A}
\end{equation}
To stress this fact we will rewrite (\ref{eqHermiteFunctions1.A}) as
\begin{equation}
W e_n (z)= F_n (\eta (z)),
\label{eqHermiteFunctions1.B}
\end{equation}
where
\begin{equation}
F_n(\eta)= \frac{(-1)^n}{\pi} L_n^0 (2 \eta^2) e^{- \eta^2},
\label{eqHermiteFunctions1.C}
\end{equation}
and
\begin{equation}
\eta (z) = |z|.
\label{eqHermiteFunctions1.D}
\end{equation}
\begin{theorem}\label{TheoremRadial1}
Let $F \in L^2 (\mathbb{R}^2)$ be a function of $\rho(z)$ only,
\begin{equation}
F(z)= G( \rho (z))
\label{eqRadial1}
\end{equation}
where $ \sqrt{\rho} G (\rho) \in L^2 (0, + \infty)$, and
\begin{equation}
\rho (z):= \left( (z-z_0) \cdot A (z-z_0) \right)^{1/2}.
\label{eqRadial2}
\end{equation}
Here $z_0 \in \mathbb{R}^2$ and $A$ is a real, symmetric, positive-definite $2 \times 2$ matrix with
\begin{equation}
\det{A}=1.
\label{eqRadial2.1}
\end{equation}
Then, we have:
\begin{equation}
F(z)= \sum_{n=0}^{\infty} \mu_n F_n (\rho (z)),
\label{eqRadial3}
\end{equation}
where $F_n $ is the Wigner function associated with $n$-th eigenstate of the harmonic oscillator given by (\ref{eqHermiteFunctions1.C}), and
\begin{equation}
\mu_n= 4 \pi(-1)^n \int_0^{\infty} G(\rho)  L_n (2 \rho^2) e^{- \rho^2} \rho d \rho.
\label{eqRadial4}
\end{equation}
\end{theorem}

\begin{proof}
A function $G$ with the conditions stated in the theorem admits the following "diagonalization" (see Section 24 of \cite{Wong} and \cite{Janssen}):
\begin{equation}
G(\rho)= \sum_{n=0}^{\infty} \mu_n F_n (\rho),
\label{eqRadial7}
\end{equation}
where the $\mu_n$ are given by (\ref{eqRadial4}), and $F_n$ is given by (\ref{eqHermiteFunctions1.C}). From (\ref{eqRadial1}) the result follows.
\end{proof}

Before we proceed, let us make the following remarks.
\begin{remark}\label{RemarkRadial1}
First of all, if we set $z_0=0$ and $A=Id$ in (\ref{eqRadial2}), then $\rho(z)=|z|$ and the function $F(z)=G(|z|)$ is radial. By considering an arbitrary positive matrix $A$, we can solve more problems, such as the one in Example \ref{Example1} below.

Secondly, let us point out that (\ref{eqRadial2.1}) does not pose any serious restriction. Indeed, suppose that $\det (A) \ne 1$. Define $B=\frac{A}{\sqrt{\det(A)}} $. Then $\det B=1$, and since
\begin{equation}
 \rho= \sqrt[4]{\det A} \widetilde{\rho} ,
\label{eqRadial7.1}
\end{equation}
where
\begin{equation}
\widetilde{\rho}= \left((z-z_0) \cdot B (z-z_0) \right)^{\frac{1}{2}},
\label{eqRadial7.2}
\end{equation}
we conclude that $F$ can also be regarded as a function of $\widetilde{\rho}$ only and the same results follow.
\end{remark}

\begin{remark}\label{RemarkRadial2}
A function with the "radial" property stated in Theorem \ref{TheoremRadial1} seems to be "diagonalized" in (\ref{eqRadial3}). Although this is true, some care is required to make this assertion. Indeed, we have to make sure that $F_n (\rho (z))$ are Wigner functions. To show that this is indeed the case, we recall the following symplectic covariance property of Wigner distributions \cite{Folland,Gosson,Grochenig,Leray,Shale,Weil,Wong} and Williamson's Theorem \cite{Williamson}. Let $Sp_2 (d)$ be the metaplectic group, i.e. the two-fold cover of $Sp(d)$. For each $S \in Sp(d)$, there exist $\pm \widetilde{S} \in Sp_2 (d)$ which project onto $S$. The metaplectic representation $Mp(d)$ is a unitary representation of $Sp_2(d)$, $ Sp_2(d) \ni \widetilde{S} \mapsto \mu (\widetilde{S})$, with the property that:
\begin{equation}
\mu (\widetilde{S})^{-1} \widehat{A} \mu (\widetilde{S}) \overset{\mathrm{Weyl}%
}{\longleftrightarrow} a \circ S ,
\label{eqMetaplectic1}
\end{equation}
for $\widehat{A}: \mathcal{S} (\mathbb{R}^d) \to \mathcal{S}^{\prime} (\mathbb{R}^d)
$ a Weyl operator with Weyl symbol $a \in \mathcal{S}^{\prime} (\mathbb{R}^{2d}) $. In particular, for Wigner functions, we have:
\begin{equation}
W \left( \mu (\widetilde{S})f ,\mu (\widetilde{S}) g \right)(z) =W (f,g)(S^{-1}z),
\label{eqMetaplectic2}
\end{equation}
for all $f,g \in \mathcal{S} (\mathbb{R}^d)$. By usual density arguments this extends to $L^2 (\mathbb{R}^d)$.

Moreover, the set of Wigner functions is left invariant under phase space translations. Altogether, if $W \psi (z)$ is a Wigner function, then under an affine symplectic transformation $W \psi (Sz -z_0)$ ($S \in Sp(d)$, $z_0 \in \mathbb{R}^{2d}$) we obtain another Wigner function.

Now, let us go back to the $2 \times 2$ matrix $A$ in (\ref{eqRadial2}). Williamson's Theorem \cite{Williamson} states that there exists $S \in Sp(1)$ and a positive number $ \lambda $ (called a Williamson invariant) such that $A= \lambda S^T S$.
By assumption $\det (A)=1$ and thus $\lambda=1$. Hence:
\begin{equation}
A= S^T S.
\label{eqMetaplectic3}
\end{equation}
We thus have that
\begin{equation}
\rho^2 (z) = (z-z_0) \cdot A (z-z_0)=  \left( S (z-z_0) \right) \cdot \left( S (z-z_0) \right)
\label{eqMetaplectic4}
\end{equation}
It follows that $F_n (\rho (z))$ is obtained from $W e_n(z)$ by the affine symplectic transformation:
\begin{equation}
z \mapsto Sz- S z_0.
\label{eqMetaplectic5}
\end{equation}
Thus $F_n (\rho (z))$ is again a Wigner function.

\end{remark}

\begin{remark}\label{RemarkRadial3}
If a function $F$ is a function of $\rho$ (\ref{eqRadial2}) only, as in Theorem \ref{TheoremRadial1}, and $F_+$ is not identically zero, then we can solve the optimization problem exactly by the following (finite) iterative procedure.

First of all notice that from Remarks \ref{RemarkRepresentations} and \ref{RemarkRadial2}, the function $F$ is diagonalized in (\ref{eqRadial3}) and thus the coefficients $\mu_n$ are in fact the eigenvalues of $\widehat{F}$, $\mathbb{F}$ and $F_\ast$. We then proceed as follows.

\vspace{0.3 cm}
\noindent
{\bf 1)} Compute the eigenvalues (\ref{eqRadial4}) until you find the first positive one, say $\mu_{k_1}$. Define
\begin{equation}
F^{(k_1)} (z):= \sum_{n=0}^{k_1} \mu_n F_n (\rho (z)).
\label{eqMetaplectic6}
\end{equation}
If
\begin{equation}
||F-F^{(k_1)}||_{L^2 (\mathbb{R}^{2d})} \le \frac{\mu_{k_1}}{\sqrt{2 \pi}},
\label{eqMetaplectic7}
\end{equation}
then we conclude that
\begin{equation}
|\mu_l| \le \sqrt{\sum_{n=k_1+1}^{+ \infty} | \mu_n|^2} < \mu_{k_1},
\label{eqMetaplectic8}
\end{equation}
for all $l=k_1+1, k_1+2, \cdots$. Consequently, $\mu_{k_1}$ is the largest eigenvalue of $\widehat{F}$ and the optimal solution is
\begin{equation}
W \psi^{(0)}(z)=  \mu_{k_1} F_{k_1} (\rho (z)).
\label{eqMetaplectic9}
\end{equation}

\vspace{0.3 cm}
\noindent
{\bf 2)} If (\ref{eqMetaplectic7}) does not hold, then look for the next positive eigenvalue $\mu_{k_2}$ $(k_2 > k_1)$ and set
\begin{equation}
F^{(k_2)} (z):= \sum_{n=0}^{k_2} \mu_n F_n (\rho (z)).
\label{eqMetaplectic10}
\end{equation}
Define $K \in \left\{k_1,k_2 \right\}$, such that
\begin{equation}
\mu_K = \max \left\{\mu_{k_1},\mu_{k_2} \right\}.
\label{eqMetaplectic11}
\end{equation}
If
\begin{equation}
||F-F^{(k_2)}||_{L^2 (\mathbb{R}^{2d})} \le \frac{\mu_K}{\sqrt{2 \pi}},
\label{eqMetaplectic12}
\end{equation}
then we conclude that
\begin{equation}
|\mu_l| \le \sqrt{\sum_{n=k_2+1}^{+ \infty} | \mu_n|^2} < \mu_K,
\label{eqMetaplectic13}
\end{equation}
for all $l=k_2+1, k_2+2, \cdots$. Consequently, $\mu_K$ is the largest eigenvalue of $\widehat{F}$ and the optimal solution is
\begin{equation}
W \psi^{(0)}(z)=  \mu_{K} F_{K} (\rho (z)).
\label{eqMetaplectic14}
\end{equation}

\vspace{0.3 cm}
\noindent
{\bf 3)} If (\ref{eqMetaplectic12}) is still not valid, then we proceed in the same fashion and obtain a set of positive eigenvalues $\mu_{k_1}, \mu_{k_2}, \cdots, \mu_{k_n}$ $(0 \le k_1 < k_2 < \cdots < k_n)$. As before, we set
\begin{equation}
F^{(k_n)} (z):= \sum_{n=0}^{k_n} \mu_n F_n (\rho (z)),
\label{eqMetaplectic15}
\end{equation}
and define $K \in \left\{k_1,k_2, \cdots, k_n \right\}$ such that \begin{equation}
\mu_K = \max \left\{\mu_{k_1},\mu_{k_2} , \cdots, \mu_{k_n} \right\}.
\label{eqMetaplectic16}
\end{equation}
If
\begin{equation}
||F-F^{(k_n)}||_{L^2 (\mathbb{R}^{2d})} \le \frac{\mu_K}{\sqrt{2 \pi}},
\label{eqMetaplectic17}
\end{equation}
then we conclude that
\begin{equation}
|\mu_l| \le \sqrt{\sum_{n=k_n+1}^{+ \infty} | \mu_n|^2} < \mu_K,
\label{eqMetaplectic18}
\end{equation}
for all $l=k_n+1, k_n+2, \cdots$. Consequently, $\mu_K$ is the largest eigenvalue of $\widehat{F}$ and the optimal solution is given by (\ref{eqMetaplectic14}).

\vspace{0.3 cm}
\noindent
Notice that condition (\ref{eqMetaplectic17}) will eventually be satisfied for some $k_n \in \mathbb{N}$, since $\sum_{n=k}^{+ \infty} | \mu_n|^2 \to 0$ as $k \to \infty$.
\end{remark}

\begin{example}\label{Example1}
A particular instance of the previous construction is a Gaussian of the form
\begin{equation}
F(z)= N \exp \left(- \alpha (z-z_0)\cdot A (z-z_0) \right),
\label{eqExample11}
\end{equation}
where $N$ and $\alpha$ are arbitrary positive constants, and where we assume that $A$ is a real, symmetric, positive-definite $2 \times 2$ matrix with $\det (A)=1$. A straightforward calculation yields for $n \ge 1$:
\begin{equation}
\begin{array}{c}
\mu_n = 4 \pi (-1)^n N \int_0^{\infty} e^{-(1+ \alpha) \rho^2} \rho L_n (2 \rho^2) d \rho=\\
\\
=  4 \pi (-1)^n N \sum_{k=0}^n \left(
\begin{array}{c}
n\\
k
\end{array}
\right) \frac{(-2)^k}{k!} \int_0^{\infty} e^{-(1+ \alpha) \rho^2} \rho^{2 k +1} d \rho =\\
\\
=\frac{(-1)^n 2 \pi N}{1 + \alpha} \sum_{k=0}^n \left(
\begin{array}{c}
n\\
k
\end{array}
\right) \left(- \frac{2}{1+ \alpha} \right)^k = \frac{(-1)^n 2 \pi N}{1 + \alpha} \left(1- \frac{2}{1+ \alpha}\right)^n= \frac{2 \pi N}{1+ \alpha} \left( \frac{1- \alpha}{1 + \alpha} \right)^n,
\end{array}
\label{eqExample12}
\end{equation}
and
\begin{equation}
\mu_0= \frac{2 \pi N}{1+ \alpha}.
\label{eqExample13}
\end{equation}
Clearly, if $\alpha =1$, then $\mu_n =\pi N \delta_{n,0}$, and the largest eigenvalue is $\mu_0$. If $\alpha <1$, then all the eigenvalues are strictly positive. Notice that in this case, the Gaussian satisfies the Robertson-Schr\"odinger uncertainty principle \cite{Narcowich}:
\begin{equation}
A^{-1} + i \alpha J \ge 0,
\label{eqExample14}
\end{equation}
where
\begin{equation}
J=\left(
\begin{array}{c c}
0 & 1\\
-1 & 0
\end{array}
\right)
\label{eqExample15}
\end{equation}
is the standard symplectic matrix. The uncertainty principle (\ref{eqExample14}) is well known to be a necessary and sufficient condition for a Gaussian measure to be the Weyl symbol of a positive trace-class operator \cite{Narconnell,Narcowich}.

Thus, if $\alpha <1$, the sequence (\ref{eqExample12}) is strictly decreasing. Hence, the largest eigenvalue is $\mu_0$.

Finally, if $\alpha >1$, then we have an alternating sequence
\begin{equation}
\mu_n= \frac{(-1)^n 2 \pi N}{1+ \alpha} \left( \frac{\alpha-1}{1 + \alpha} \right)^n, \hspace{1 cm} n \ge 0.
\label{eqExample16}
\end{equation}
But again, since the moduli sequence
\begin{equation}
|\mu_n|= \frac{2 \pi N}{1+ \alpha} \left( \frac{\alpha-1}{1 + \alpha} \right)^n
\label{eqExample17}
\end{equation}
is strictly decreasing, $\mu_0$ is again the largest positive eigenvalue. Consequently, for any $\alpha >0$, the Wigner function closest to the Gaussian measure (\ref{eqExample11}) is:
\begin{equation}
W \psi^{(0)} (z) = \frac{2 \pi  N}{1 + \alpha} F_0 (\rho (z)) = \frac{ 2  N}{1 + \alpha} e^{-  (z-z_0)\cdot A (z-z_0) }.
\label{eqExample18}
\end{equation}
Thus in particular, if we have $F(z)= We_0 (z)= \frac{1}{\pi} e^{- |z|^2}$ (that is: $N=\frac{1}{\pi}$, $\alpha =1$, $A=Id$ and $z_0=0$), then we obtain $W \psi^{(0)} (z) =We_0 (z)$ as expected.
\end{example}

\begin{remark}\label{RemarkToeplitzOps}
Before we conclude this section, we remark that the optimization problem considered in this paper is intimately related with the so-called localization or Toeplitz operators in time-frequency analysis \cite{Daubechies,Flandrin,Ramanathan,Slepian1,Slepian2} and quantum mechanics \cite{Bracken,Lieb}. In \cite{Bracken} the authors addressed the optimization problems
\begin{equation}
\mbox{min}_{\psi \in L^2 (\mathbb{R}) \backslash \left\{ 0 \right\}} \frac{\int_D W \psi (z) dz}{\int_{\mathbb{R}^2} W \psi (z) dz}
\label{eqRemarkToeplitzOps1}
\end{equation}
and
\begin{equation}
\mbox{max}_{\psi \in L^2 (\mathbb{R})\backslash \left\{ 0 \right\}} \frac{\int_D W \psi (z) dz}{\int_{\mathbb{R}^2} W \psi (z) dz}
\label{eqRemarkToeplitzOps2}
\end{equation}
where $D \subset \mathbb{R}$ is some bounded domain whose boundary $\partial D$ is a regular curve.

If we define $F(z)= \chi_D(z) \in L^2 (\mathbb{R}^2)$ to be the characteristic function of $D$, then it is straightforward to prove that the Wigner function $W \psi^{(0)}$ closest in $L^2$ to $F$ is the optimal solution of (\ref{eqRemarkToeplitzOps2}).

In \cite{Bracken,Flandrin,Lieb} the authors proved that Gaussians maximize (\ref{eqRemarkToeplitzOps2}), when $D$ is a disk, a poly-disk or a ball.
\end{remark}

\section{The general approximation procedure}

A crucial point in our derivation will be the Courant-Fischer min-max theorem, which we recapitulate here for completeness.

\begin{theorem}\label{TheoremCourantFischer}
{\bf (Courant-Fischer min-max theorem)} Let $A$ be an $N \times N$ hermitian matrix and write its eigenvalues as a decreasing sequence $\alpha_1 \ge \alpha_2  \ge \cdots \ge \alpha_N$. Then we have:
\begin{equation}
\alpha_j =
sup_{dim(V)=j} ~ inf_{v \in V, ||v||=1} ~  \overline{v} \cdot A v,
\label{eqTheoremCourantFischer1}
\end{equation}
and
\begin{equation}
\alpha_j = inf_{dim(V)=N-j+1}~ sup_{v \in V, ||v||=1} ~ \overline{v} \cdot A v,
\label{eqTheoremCourantFischer2}
\end{equation}
for all $j=1,2, \cdots, N$ and $V$ ranges over all subspaces of $\mathbb{C}^N$ with the indicated dimension. Here $||v||^2 = | v_1|^2 + \cdots |v_N|^2$.
\end{theorem}

Before we proceed, let us recall that $\left\{\lambda_j \right\}_j$ are the distinct eigenvalues of the matrix $\mathbb{F}$, whereas $\left\{\mu_j \right\}_j$ are its eigenvalues with multiplicities. With our previous notation (\ref{eq17}-\ref{eq17.2}) the eigenvalues of $\mathbb{F}^{(N)}$ are $\left\{\mu_j^{(N)} \right\}_j$. If we rearrange the eigenvalues of $\mathbb{F}^{(N)}$ as a decreasing sequence $\left\{\alpha_j^{(N)} \right\}_j$ (as described in the Courant-Fischer Theorem), we may write:
\begin{equation}
\alpha_j^{(N)}= \left\{
\begin{array}{l l}
\mu_j^{(N)} , & j=1, \cdots, N_+\\
0 , & j=N_+ +1, \cdots, N_+ +N_K\\
\mu_{j-N-1}^{(N)} , & j=N_+ + N_K + 1, \cdots, N
\end{array}
\right.
\label{eqTheoremCourantFischer2.1}
\end{equation}

We follow closely \cite{Koehler} for the estimates of the eigenvalues. We just have to make some adaptations to complex-valued matrices and to the facts that we have the $L^2$-bound (\ref{eq16}) and that the matrix $\mathbb{F}$ is not necessarily positive.

We start by proving that each sequence $\left\{\mu_j^{(N)} \right\}_{N \in \mathbb{N}}$ and $\left\{- \mu_{-j}^{(N)} \right\}_{N \in \mathbb{N}}$ with fixed $j \ge 1$ of non-zero eigenvalues is non-decreasing with respect to $N$ and, since they are bounded, they are convergent.

But before we do that, we remark that we may at various moments denote the vectors $(x_1, \cdots, x_N) \in \mathbb{C}^N$ and $(x_1, \cdots, x_N , 0,0,0, \cdots) \in l^2 (\mathbb{N})$ by the same symbol $x^{(N)}$ according to our convenience.

\begin{proposition}\label{PropositionSequence}
With the previous notation, we have that for $1 \le j \le N_+$
\begin{equation}
\mu_j^{(N)} \nearrow \mu_j, \hspace{1 cm} N \to \infty,
\label{eqPropositionSequence1}
\end{equation}
while
\begin{equation}
\mu_{-k}^{(N)} \searrow \mu_{- k}, \hspace{1 cm} N \to \infty,
\label{eqPropositionSequence2}
\end{equation}
for $1 \le k \le N_-$.
\end{proposition}

\begin{proof}
By the Courant-Fischer theorem, we have for $1 \le j \le N$:
\begin{equation}
\alpha_j^{(N)} = sup_{dim V = j} ~ inf_{x^{(N)} \in V, ||x^{(N)}||=1} ~ \overline{x^{(N)}} \cdot \mathbb{F}^{(N)} x^{(N)},
\label{eqPropositionSequence3}
\end{equation}
and
\begin{equation}
\alpha_j^{(N+1)} = sup_{dim V = j} ~ inf_{x^{(N+1)} \in V, ||x^{(N+1)}||=1} ~ \overline{x^{(N+1)}} \cdot \mathbb{F}^{(N+1)} x^{(N+1)}.
\label{eqPropositionSequence4}
\end{equation}
Notice that we can rewrite (\ref{eqPropositionSequence3}) as:
\begin{equation}
\alpha_j^{(N)} = sup_{dim V^{\prime} = j} ~ inf_{x^{(N+1)} \in V^{\prime}, ||x^{(N+1)}||=1} ~ \overline{x^{(N+1)}} \cdot \mathbb{F}^{(N+1)} x^{(N+1)},
\label{eqPropositionSequence5}
\end{equation}
where the supremum is taken over all subspaces $V^{\prime}$ of $\mathbb{C}^{N+1}$ with dimension $j$, such that $x^{(N+1)} \in V^{\prime}$ if and only if its $(N+1)$-th coordinate is zero. From (\ref{eqPropositionSequence4}) and (\ref{eqPropositionSequence5}) it is then obvious that
\begin{equation}
\alpha_j^{(N)} \le \alpha_j^{(N+1)}.
\label{eqPropositionSequence6}
\end{equation}
Thus, for fixed $j$ this is a non-decreasing sequence. In particular, we have:
\begin{equation}
(N+1)_+ \ge N_+
\label{eqPropositionSequence7}
\end{equation}
Thus for $1 \le j \le N_+$, we have:
\begin{equation}
\mu_j^{(N)}=\alpha_j^{(N)} \le \alpha_j^{(N+1)} = \mu_j^{(N+1)}.
\label{eqPropositionSequence8}
\end{equation}
Next, if we replace $\mathbb{F}^{(N)}$ by $- \mathbb{F}^{(N)}$, then: $\alpha_j^{(N)} \to - \alpha_{N-j+1}^{(N)}$. If we then apply the previous conclusions to $-\mathbb{F}^{(N)}$, it follows that:
\begin{equation}
(N+1)_- \ge N_-,
\label{eqPropositionSequence9}
\end{equation}
and
\begin{equation}
\mu_{- j}^{(N)} \ge \mu_{-j}^{(N+1)}
\label{eqPropositionSequence10}
\end{equation}
for $1 \le j \le N_-$.

Concerning convergence, we invoke a familiar theorem for the spectral radius of bounded linear operators on Banach spaces. For any element $\mu^{(N)}$ in the spectrum of $\mathbb{F}^{(N)}$ and for fixed $N \in \mathbb{N}$, we have:
\begin{equation}
|\mu^{(N)}| \le ||\mathbb{F}^{(N)}||_{Op} \le ||\mathbb{F}^{(N)}||_{l^2} \le ||\mathbb{F}||_{l^2}=(2 \pi)^{d/2} ||F||_{L^2 (\mathbb{R}^{2d})}.
\label{eqPropositionSequence11}
\end{equation}
We conclude that the sequences $\left\{\mu_j^{(N)} \right\}_{N \in \mathbb{N}}$ and $\left\{\mu_{-k}^{(N)} \right\}_{N \in \mathbb{N}}$ for fixed $j,k$ are bounded and monotone, and hence convergent.

It remains to prove that the limit of $\mu_j^{(N)}$, for fixed $j \ne 0$, is $\mu_j$ as $N \to \infty$.

We prove the result for the positive eigenvalues $j \ge 1$. The proof for negative eigenvalues is identical.

So assume that the limit as $N \to \infty$ of some $\mu_j^{(N)}$ is not in the spectrum of $\mathbb{F}$. Let $j_0$ be the smallest $j \ge 1$ for which $\mu_j^{(N)} \to \mu_j^{\ast} $, and $\mu_j^{\ast} $ is not in the spectrum of $\mathbb{F}$. Recalling the notation (\ref{eq17.2.0}-\ref{eq17.4}) for the orthonormal eigenvectors of $\mathbb{F}$ and $\mathbb{F}^{(N)}$, we have:
\begin{equation}
\begin{array}{c}
||\mathbb{F} c_{j_0}^{(N)} - \mu_{j_0}^{\ast} c_{j_0}^{(N)}||_{l^2} = ||(\mathbb{F}-\mathbb{F}^{(N)}) c_{j_0}^{(N)} +\mathbb{F}^{(N)} c_{j_0}^{(N)}- \mu_{j_0}^{\ast} c_{j_0}^{(N)}||_{l^2} \le \\
\\
\le ||(\mathbb{F}-\mathbb{F}^{(N)}) c_{j_0}^{(N)}||_{l^2}+ || (\mu_{j_0}^{(N)}- \mu_{j_0}^{\ast}) c_{j_0}^{(N)}||_{l^2} \le  \\
\\
\le ||(\mathbb{F}-\mathbb{F}^{(N)}) ||_{Op} ||c_{j_0}^{(N)}||_{l^2}+ | \mu_{j_0}^{(N)}- \mu_{j_0}^{\ast}| ~||c_{j_0}^{(N)}||_{l^2} \le  \\
\\
\le (2 \pi)^{d/2}||F-F^{(N)})||_{L^2 (\mathbb{R}^{2d})} + |\mu_{j_0}^{(N)}- \mu_{j_0}^{\ast}|
\end{array}
\label{eqTheoremEstimatesEigenvectorsA.1.0}
\end{equation}
And thus:
\begin{equation}
||\mathbb{F} c_{j_0}^{(N)} - \mu_{j_0}^{\ast} c_{j_0}^{(N)}||_{l^2} \to 0,
\label{eqTheoremEstimatesEigenvectorsA.1}
\end{equation}
as $N \to \infty$.

Next we expand $c_{j_0}^{(N)}$ in the basis $\mathcal{B}$ (\ref{eq17.3}) of $l^2 (\mathbb{N})$:
\begin{equation}
c_{j_0}^{(N)} = \sum_{\alpha \in \mathcal{A}} a_{\alpha} e_{\alpha},
\label{eqTheoremEstimatesEigenvectorsB}
\end{equation}
with
\begin{equation}
||c_{j_0}^{(N)} ||_{l^2}= \sum_{\alpha \in \mathcal{A}} |a_{\alpha}|^2=1.
\label{eqTheoremEstimatesEigenvectorsC}
\end{equation}
On the other hand, if $\mu_{j_0}^{\ast}$ is not in $Spec (\mathbb{F})$ - the spectrum of $\mathbb{F}$ - then
\begin{equation}
0 < K_{j_0}:= inf \left\{ |\mu_{j_0}^{\ast} - \eta |: ~ \eta \in  Spec (\mathbb{F})\right\}.
\label{eqTheoremEstimatesEigenvectorsB.1}
\end{equation}
It follows from (\ref{eqTheoremEstimatesEigenvectorsB}) that
\begin{equation}
||\left(\mathbb{F} - \mu_{j_0}^{\ast} Id \right) c_{j_0}^{(N)} ||_{l^2} =||\left(\mathbb{F} - \mu_{j_0}^{\ast} Id \right) \sum_{\alpha \in \mathcal{A}} a_{\alpha}e_{\alpha}  ||_{l^2} = ||\sum_{\alpha \in \mathcal{A}} a_{\alpha} (\beta_{\alpha} - \mu_{j_0}^{\ast}) e_{\alpha} ||_{l^2}
\label{eqTheoremEstimatesEigenvectorsB.2}
\end{equation}
where $\beta_{\alpha}$ is the eigenvalue of $\mathbb{F}$ associated with the eigenvector $e_{\alpha}$.

From Pithagoras' theorem and (\ref{eqTheoremEstimatesEigenvectorsC},\ref{eqTheoremEstimatesEigenvectorsB.1},\ref{eqTheoremEstimatesEigenvectorsB.2}), we have:
\begin{equation}
||\left(\mathbb{F} - \mu_{j_0}^{\ast} Id \right) c_{j_0}^{(N)} ||_{l^2} = \sqrt{ \sum_{\alpha \in \mathcal{A}} |a_{\alpha}|^2 |\beta_{\alpha} - \mu_{j_0}^{\ast}|^2} >K_{j_0} \sqrt{ \sum_{\alpha \in \mathcal{A}} |a_{\alpha}|^2 }= K_{j_0}>0,
\label{eqTheoremEstimatesEigenvectorsB.3}
\end{equation}
 which contradicts (\ref{eqTheoremEstimatesEigenvectorsA.1}).

 This proves that $\mu_j^{\ast} = \lim_{N \to \infty} \mu_j^{(N)}$ is in the spectrum of $\mathbb{F}$ for all $j \ge 1$. But it still remains to prove that $\mu_j^{\ast} =\mu_j$.

 Suppose that $\bar{j}$ is the smallest $j \ge 1$, such that $\mu_j^{(N)} \to \mu_j^{\ast} \ne \mu_j$. From this assumption and the monotonicity, we must have in fact $\mu_{\bar{j}}^{\ast} < \mu_{\bar{j}}$. From the min-max principle and the monotonicity, we have for $N > \bar{j}$:
 \begin{equation}
\begin{array}{c}
\mu_{\bar{j}}^{(N)} = \mbox{sup}_{dim V=\bar{j}} ~ \mbox{inf}_{x^{(N)} \in V; x^{(N)} \ne 0} \frac{\overline{x^{(N)} } \cdot \mathbb{F}^{(N)} x^{(N)}}{||x^{(N)}||^2}=\\
\\
=\mbox{sup}_{dim V^{\prime}=\bar{j}} ~\mbox{inf}_{x^{(N)} \in V^{\prime}; x^{(N)} \ne 0} \frac{\overline{x^{(N)}} \cdot \mathbb{F} x^{(N)}}{||x^{(N)}||_{l^2}^2} \le \mu_{\bar{j}}^{\ast} < \mu_{\bar{j}}
\end{array}
\label{eqTheoremEstimatesEigenvectorsB.4}
\end{equation}
where $V$ ranges over all $\bar{j}$-dimensional subspaces of $\mathbb{C}^N$ and $V^{\prime}$ ranges over all $\bar{j}$-dimensional subspaces of $l^2 (\mathbb{N})$ such that $x^{(N)} =(x_1, x_2, \cdots) \in V^{\prime}$ implies $x_{N+1}=x_{N+2}=x_{N+3}= \cdots =0$.

But as $N \to \infty$, the spaces $V^{\prime}$ become dense in $l^2 (\mathbb{N})$ and, from the last inequality in (\ref{eqTheoremEstimatesEigenvectorsB.4}), we have a contradiction with:
\begin{equation}
\mu_{\bar{j}}=\mbox{sup}_{dim W=\bar{j}} ~ \mbox{inf}_{x \in W; x \ne 0} \frac{\overline{x} \cdot \mathbb{F} x}{||x||_{l^2}^2}
\label{eqTheoremEstimatesEigenvectorsB.5}
\end{equation}
where $W$ ranges over all  $\bar{j}$-dimensional subspaces of $l^2 (\mathbb{N})$.
\end{proof}

In the next theorem, we obtain an estimate for the approximate eigenvalues. In particular the estimate stated in Eq.(\ref{eqTheoremEstimatesEigenvalues3}) can be regarded as an infinite dimensional version of the Weyl \cite{Weyl} or the Wielandt-Hoffman \cite{Tao} inequalities. In \cite{Balazs1,Balazs2} the authors considered the multiplication of a fixed multiplier pattern (Bessel multiplier), which is inserted between the analysis and the synthesis operator. They then considered the perturbation of Bessel sequences and obtain results which have some resemblance with our next theorem.

\begin{theorem}\label{TheoremEstimatesEigenvalues}
With the assumption (\ref{eq16}), we have:
\begin{equation}
\mu_j^{(N)} \le \mu_j \le max \left\{ \mu_j^{(N)} +(2 \pi)^{d/2} \epsilon ||F||_{L^2 (\mathbb{R}^{2d})}, 2 (2 \pi)^{d/2} \epsilon ||F||_{L^2 (\mathbb{R}^{2d})} \right\},
\label{eqTheoremEstimatesEigenvalues1}
\end{equation}
for $j=1,2, \cdots, N_+$, and
\begin{equation}
min \left\{ - \mu_{-j}^{(N)} - (2 \pi)^{d/2} \epsilon ||F||_{L^2 (\mathbb{R}^{2d})}, - 2 (2 \pi)^{d/2} \epsilon ||F||_{L^2 (\mathbb{R}^{2d})}  \right\} \le \mu_{-j} \le \mu_{-j}^{(N)},
\label{eqTheoremEstimatesEigenvalues2}
\end{equation}
for $j=1,2, \cdots, N_-$. Consequently,
\begin{equation}
|\mu_j^{(N)} - \mu_j| < 2 (2 \pi)^{d/2} \epsilon ||F||_{L^2 (\mathbb{R}^{2d})},
\label{eqTheoremEstimatesEigenvalues3}
\end{equation}
for all $j=-N_{-}, \cdots,-1,1, \cdots, N_+$.
\end{theorem}

\begin{proof}
Choose $j \le N <M$. Again, by the min-max theorem:
\begin{equation}
\alpha_j^{(M)} = inf_{dim V= M - j+1} ~ sup_{x^{(M)} \in V \backslash \left\{ 0 \right\}} ~ \frac{\overline{x^{(M)}} \cdot \mathbb{F}^{(M)}
x^{(M)}}{||x^{(M)}||^2}.
\label{eqTheoremEstimatesEigenvalues4}
\end{equation}
For $x^{(M)}=(x_1, \cdots,x_N, \cdots, x_M)$, set
\begin{equation}
y_1 = \left( \sum_{i=1}^N |x_i|^2 \right)^{1/2} , \hspace{1 cm} y_2 = \left( \sum_{i=N+1}^M |x_i|^2, \right)^{1/2}
\label{eqTheoremEstimatesEigenvalues5}
\end{equation}
so that
\begin{equation}
||x^{(M)}||^2  = \sum_{i=1}^M |x_i|^2 = y_1^2 + y_2^2 .
\label{eqTheoremEstimatesEigenvalues6}
\end{equation}
Moreover,
\begin{equation}
\begin{array}{c}
\overline{x^{(M)}} \cdot \mathbb{F}^{(M)} x^{(M)}= \sum_{i,k=1}^M \overline{x_i} f_{ik} x_k = \\
\\
= \overline{x^{(N)}} \cdot \mathbb{F}^{(N)} x^{(N)} + 2 Re\left(\sum_{i=1}^N \sum_{k=N+1}^M \overline{x_i} f_{ik} x_k \right)+ \sum_{i,k=N+1}^M \overline{x_i} f_{ik} x_k
\end{array}
\label{eqTheoremEstimatesEigenvalues7}
\end{equation}
Applying twice the Cauchy-Schwarz inequality to the second term on the right-hand side of the previous equation:
\begin{equation}
\begin{array}{c}
2 Re\left(\sum_{i=1}^N \sum_{k=N+1}^M \overline{x_i} f_{ik} x_k \right) \le 2 \sum_{i=1}^N |x_i|~\left|\sum_{k=N+1}^M  f_{ik} x_k \right| \le \\
\\
\le 2 y_1 \left( \sum_{i=1}^N \left|\sum_{k=N+1}^M  f_{ik} x_k \right|^2  \right)^{1/2} \le \\
\\
\le 2 y_1 \left[ y_2^2 \sum_{i=1}^N \left(\sum_{k=N+1}^M |f_{ik}|^2 \right) \right]^{1/2}= \\
\\
=2 y_1 y_2 \left( \sum_{i=1}^N \sum_{k=N+1}^M |f_{ik}|^2 \right)^{1/2} \le 2 \epsilon_N y_1 y_2 ||F||_{L^2 (\mathbb{R}^{2d})},
\end{array}
\label{eqTheoremEstimatesEigenvalues8}
\end{equation}
where
\begin{equation}
\epsilon_N := \frac{\left( \sum_{i=1}^N \sum_{k=N+1}^{\infty} |f_{ik}|^2 \right)^{1/2}}{||F||_{L^2 (\mathbb{R}^{2d})}}.
\label{eqTheoremEstimatesEigenvalues8.1}
\end{equation}
Next we consider the third term on the right-hand side of (\ref{eqTheoremEstimatesEigenvalues7}):
\begin{equation}
\sum_{i,k=N+1}^M \overline{x_i} f_{ik} x_k \le y_2^2 \left|\frac{\sum_{i,k=N+1}^M \overline{x_i} f_{ik} x_k}{\sum_{l=N+1}^M |x_l|^2} \right| \le \rho_N y_2^2,
\label{eqTheoremEstimatesEigenvalues9}
\end{equation}
where
\begin{equation}
\rho_N:=sup_{x} \left| \frac{\sum_{i,k=N+1}^{\infty} \overline{x_i} f_{ik} x_k}{\sum_{l=N+1}^{\infty} |x_l|^2} \right|
\label{eqTheoremEstimatesEigenvalues10}
\end{equation}
and the supremum is taken over all $x=(0, \cdots, 0, x_{N+1}, x_{N+2}, \cdots ) \in l^2 (\mathbb{N})$.

Finally consider the first term in (\ref{eqTheoremEstimatesEigenvalues7}):
\begin{equation}
\overline{x^{(N)}} \cdot \mathbb{F}^{(N)} x^{(N)}= y_1^2 \left(\frac{\overline{x^{(N)}} \cdot \mathbb{F}^{(N)} x^{(N)}}{\sum_{i=1}^N | x_i|^2} \right) .
\label{eqTheoremEstimatesEigenvalues11}
\end{equation}
Recall from (\ref{eqTheoremEstimatesEigenvalues4}) that we are considering planes of dimension $M-j+1$. Their codimension in $\mathbb{C}^M$ is $j-1$. But, if we set the coordinates $x_{N+1}= \cdots =x_M=0$ in these planes, then the resulting codimension in $\mathbb{C}^N$ is lower or equal to $j-1$. For some planes it is strictly smaller. We conclude that the set of all subspaces of $\mathbb{C}^N$ with dimension $N-j+1$ is a proper subset of the set of subspaces of $\mathbb{C}^N$ with codimension $j-1$ obtained in this fashion. It follows by the min-max principle that
\begin{equation}
\mbox{inf}_{dim V^{\prime}=M-j+1} ~\mbox{sup}_{x^{(M)} \in V^{\prime} \backslash \left\{0 \right\} } \frac{\overline{x^{(N)}} \cdot \mathbb{F}^{(N)} x^{(N)}}{\sum_{i=1}^N |x_i|^2}  \le  \alpha_j^{(N)},
\label{eqTheoremEstimatesEigenvalues11.1}
\end{equation}
where, as before, if $x^{(N)} =(x_1, x_2, \cdots , x_M) \in V^{\prime}$, then $x_{N+1}= \cdots=x_M=0$.

Consequently from (\ref{eqTheoremEstimatesEigenvalues4},\ref{eqTheoremEstimatesEigenvalues7},\ref{eqTheoremEstimatesEigenvalues8},
\ref{eqTheoremEstimatesEigenvalues9},\ref{eqTheoremEstimatesEigenvalues11},\ref{eqTheoremEstimatesEigenvalues11.1}) we have that:
\begin{equation}
\alpha_j^{(M)} \le max_{(y_1,y_2) \in (\mathbb{R}^+)^2}  \frac{y_1^2 \alpha_j^{(N)} + 2 \epsilon_N ||F||_{L^2 (\mathbb{R}^{2d})} y_1 y_2 + \rho_N y_2^2}{y_1^2 + y_2^2}
\label{eqTheoremEstimatesEigenvalues12}
\end{equation}
In particular, for the sequence of positive eigenvalues $(j=1, \cdots, N_+)$, we have:
\begin{equation}
\mu_j^{(M)} \le max_{(y_1,y_2) \in (\mathbb{R}^+)^2}  \frac{y_1^2 \mu_j^{(N)} + 2 \epsilon_N ||F||_{L^2 (\mathbb{R}^{2d})} y_1 y_2 + \rho_N y_2^2}{y_1^2 + y_2^2}
\label{eqTheoremEstimatesEigenvalues12.1}
\end{equation}
The maximum on the right-hand side of (\ref{eqTheoremEstimatesEigenvalues12.1}) is easily computed and we obtain:
\begin{equation}
\mu_j^{(M)} \le \frac{\mu_j^{(N)}+ \rho_N + \sqrt{(\mu_j^{(N)}- \rho_N)^2 + 4 \epsilon_N^2 ||F||_{L^2 (\mathbb{R}^{2d})}^2} }{2}
\label{eqTheoremEstimatesEigenvalues13}
\end{equation}
Since the right-hand side is independent of $M >N$, we have from Proposition \ref{PropositionSequence}:
\begin{equation}
\begin{array}{c}
\mu_j^{(N)} \le \mu_j \le \frac{\mu_j^{(N)}+ \rho_N + \sqrt{(\mu_j^{(N)}- \rho_N)^2 + 4 \epsilon_N^2 ||F||_{L^2 (\mathbb{R}^{2d})}^2} }{2} \le \\
\\
\le max \left\{\mu_j^{(N)} +\epsilon_N ||F||_{L^2 (\mathbb{R}^{2d})} , \rho_N+\epsilon_N ||F||_{L^2 (\mathbb{R}^{2d})} \right\},
\end{array}
\label{eqTheoremEstimatesEigenvalues14}
\end{equation}
where we used the inequality $\sqrt{a^2+ b^2} \le |a|+ |b|$. If we apply the Cauchy-Schwartz inequality twice as in (\ref{eqTheoremEstimatesEigenvalues8}), we conclude that
\begin{equation}
\left| \frac{\sum_{i,k=n+1}^{\infty} \overline{x_i} f_{ik} x_k}{\sum_{l=n+1}^{\infty} |x_l|^2} \right| \le \left(\sum_{i,k=n+1}^{\infty} |f_{ik}|^2 \right)^{1/2}
\label{eqTheoremEstimatesEigenvalues14.1}
\end{equation}
Thus, in particular:
\begin{equation}
\rho_N \le \left(\sum_{i,k=N+1}^{\infty} |f_{ik}|^2 \right)^{1/2} \le  (2 \pi)^{d/2}||F-F^{(N)}||_{L^2 (\mathbb{R}^{2d})} < (2 \pi)^{d/2} \epsilon ||F||_{L^2 (\mathbb{R}^{2d})}
\label{eqTheoremEstimatesEigenvalues14.2}
\end{equation}
Likewise:
\begin{equation}
\epsilon_N \le (2 \pi)^{d/2} \frac{||F-F^{(N)}||_{L^2 (\mathbb{R}^{2d})}}{||F||_{L^2 (\mathbb{R}^{2d})}}< (2 \pi)^{d/2} \epsilon
\label{eqTheoremEstimatesEigenvalues15}
\end{equation}
From (\ref{eqTheoremEstimatesEigenvalues14},\ref{eqTheoremEstimatesEigenvalues14.2},\ref{eqTheoremEstimatesEigenvalues15}), we recover (\ref{eqTheoremEstimatesEigenvalues1}).

The result for the negative eigenvalues can be easily obtained by considering the matrix $- \mathbb{F}^{(N)}$ as before.
\end{proof}

We now turn to the eigenspaces. In the next theorem, we denote by
\begin{equation}
dist (b,A):= inf \left\{||b-a||_{l^2}: ~ a \in A \right\}
\label{eqTheoremEstimatesEigenvalues16}
\end{equation}
the distance of $b \in l^2 (\mathbb{N})$ to the set $A \subset l^2 (\mathbb{N})$.

We also make the following observation. For fixed $j \in \left\{-N_-, \cdots, -1, 1, \cdots, N \right\}$, let
\begin{equation}
m_j^{(N)}:= min \left\{|\mu_j^{(N)}|, ~|\mu_k^{(N)}- \mu_j^{(N)} |: ~  \mu_k^{(N)} \ne \mu_j^{(N)} \right\}
\label{eqObservation1}
\end{equation}
Thus, $m_j^{(N)}$ measures the distance between $\mu_j^{(N)}$ and the eigenvalue of $\mathbb{F}^{(N)}$ closest to it. Of course, as $N \to \infty$, this converges to
\begin{equation}
m_j:= min \left\{|\mu_j|, ~|\mu_k- \mu_j |: ~  \mu_k \ne \mu_j \right\} >0.
\label{eqObservation2}
\end{equation}
Clearly, for fixed $j$, we may choose $\epsilon>0$ sufficiently small and $N= N( \epsilon)$ sufficiently large so that (\ref{eq16}) holds and:
\begin{equation}
\epsilon < \frac{m_j^{(N)}}{4 (2 \pi)^d ||F||_{L^2 (\mathbb{R}^{2d})}}.
\label{eqObservation3}
\end{equation}
We also remark that, even though the eigenvalues $\mu_j^{(N)}$ converge to $\mu_j$ as $N \to \infty$, the eigenvectors $c_j^{(N)}$ need not converge in $l^2$ to $c_j$. This is because the eigenspace $\mathcal{H}_j$ of $\mu_j$ may have dimension $n_j >1$. Thus $c_j^{(N)}$ may converge to some other eigenvector with the same eigenvalue $\mu_j$ other than $c_j$ or it may not even converge at all. However, what does happen is that its distance to the eigenspace $\mathcal{H}_j$ tends to zero. We also note that the eigenvalues converge in a uniform way. By this we mean that the estimates in (\ref{eqTheoremEstimatesEigenvalues3}) are independent of $j$. On the contrary our estimates in the next theorem for the eigenvector are not uniform. But since we are interested in estimating only one eigenvalue (the largest) and the corresponding eigenspace, that is fine.

\begin{theorem}\label{TheoremEstimatesEigenvectors1}
For fixed $j \in \mathbb{Z} \backslash \left\{0 \right\}$, choose  $ \epsilon $ and $N=N ( \epsilon) \in \mathbb{N}$ such that for the truncated matrix $\mathbb{F}^{(N)} $ we have
\begin{equation}
j \in \left\{ - N_-, \cdots, -1, 1, \cdots, N_+ \right\},
\label{eqObservation4.0}
\end{equation}
and
\begin{equation}
0 < \frac{||F-F^{(N)}||_{L^2 (\mathbb{R}^{2d})}}{||F||_{L^2 (\mathbb{R}^{2d})}}
 < \epsilon < \frac{m_j^{(N)}}{4 (2 \pi)^d ||F||_{L^2 (\mathbb{R}^{2d})}}.
\label{eqObservation4}
\end{equation}
Define
\begin{equation}
M_j^{(N)}:= m_j^{(N)} - 4 (2 \pi)^d \epsilon ||F||_{L^2 (\mathbb{R}^{2d})} >0.
\label{eqObservation5}
\end{equation}
Let $\mathcal{H}_j$ denote the eigenspace of $\mathbb{F}$ associated with the eigenvalue $\mu_j$, and let $c_j^{(N)}$ be a normalized eigenvector of $\mathbb{F}^{(N)}$ associated with the eigenvalue $\mu_j^{(N)}$. We then have:
\begin{equation}
dist \left(c_j^{(N)}, \mathcal{H}_j \right) < \frac{3 (2 \pi)^{d/2}\epsilon || F||_{L^2 (\mathbb{R}^{2d})}}{M_j^{(N)}}.
\label{eqTheoremEstimatesEigenvalues17}
\end{equation}
\end{theorem}

\begin{proof}
From eq. (\ref{eqTheoremEstimatesEigenvectorsA.1.0}) with $j_0$ and $\mu_{j_0}^{\ast}$ replaced by $j$ and $\mu_j$, we have:
\begin{equation}
||\mathbb{F} c_j^{(N)} - \mu_j c_j^{(N)}||_{l^2} \le 3(2 \pi)^{d/2} \epsilon ||F||_{L^2 (\mathbb{R}^{2d})}
\label{eqA}
\end{equation}
where we used (\ref{eqTheoremEstimatesEigenvalues3}).

This equation is roughly equivalent to saying that eigenfunctions of $F^{(N)}$ become $3(2 \pi)^{d/2}\epsilon$-pseudoeigenfunctions of the
operator $F$ with pseudoeigenvalue $\mu_j$ (pseudospectra is more often associated
with non-normal operators \cite{Trefethen}, but it can be also useful in the analysis of normal operators. There are several definitions and the above has been introduced by Landau inside the proof of a Szeg\"o theorem in \cite{Landau}. Similar heuristics have been used in \cite{Abreu2}).

As in (\ref{eqTheoremEstimatesEigenvectorsB}), we expand $c_j^{(N)}$ in the orthonormal basis $\mathcal{B}$ (\ref{eq17.3}) of $l^2 (\mathbb{N})$ formed by the eigenvectors of $\mathbb{F}$:
\begin{equation}
c_j^{(N)} = \sum_{\alpha \in \mathcal{A}} a_{\alpha} e_{\alpha},
\label{eqTheoremEstimatesEigenvectorsB.1}
\end{equation}
with
\begin{equation}
||c_j^{(N)} ||_{l^2}= \sum_{\alpha \in \mathcal{A}} |a_{\alpha}|^2=1.
\label{eqTheoremEstimatesEigenvectorsC.1}
\end{equation}
Let $\mathcal{A}_j \subset \mathcal{A}$ denote the set of indices such that
\begin{equation}
\mathcal{H}_j = Span \left\{ e_{\alpha}: ~\alpha \in \mathcal{A}_j \right\}.
\label{eqTheoremEstimatesEigenvectorsD}
\end{equation}
These are the indices associated with the eigenvectors in the basis $\mathcal{B}$ which have eigenvalue $\mu_j$. Clearly, if $\mu_j$ has multiplicity one, then $\mathcal{A}_j=\left\{j \right\}$, if it has multiplicity two, then $\mathcal{A}_j=\left\{j-1, j \right\}$ or $\mathcal{A}_j=\left\{j,j+1 \right\}$, \textit{etc}. We also denote by $\mathcal{A}_j^c= \mathcal{A} \backslash \mathcal{A}_j$ its complement.

We may thus write $c_j^{(N)}$ as
\begin{equation}
c_j^{(N)} =c_{j,\|}^{(N)}+ c_{j,\bot }^{(N)},
\label{eqTheoremEstimatesEigenvectorsE}
\end{equation}
where
\begin{equation}
c_{j,\|}^{(N)}:= \sum_{\alpha \in \mathcal{A}_j} a_{\alpha} e_{\alpha} \in \mathcal{H}_j,
\label{eqTheoremEstimatesEigenvectorsF}
\end{equation}
and
\begin{equation}
c_{j,\bot}^{(N)}:= \sum_{\alpha \in \mathcal{A}_j^c} a_{\alpha} e_{\alpha} \in \mathcal{H}_j^{\bot}.
\label{eqTheoremEstimatesEigenvectorsG}
\end{equation}

From (\ref{eqA},\ref{eqTheoremEstimatesEigenvectorsE}-\ref{eqTheoremEstimatesEigenvectorsG}) it follows that
\begin{equation}
\begin{array}{c}
3(2 \pi)^{d/2} \epsilon ||F||_{L^2 (\mathbb{R}^{2d})} > ||(\mathbb{F} - \mu_j Id)\left( c_{j,\|}^{(N)}+ c_{j,\bot }^{(N)}\right) ||_{l^2}=\\
\\
=||(\mathbb{F} - \mu_j Id) c_{j,\bot }^{(N)}||_{l^2}= ||\sum_{\alpha \in \mathcal{A}_j^c} a_{\alpha} (\beta_{\alpha}-\mu_j) e_{\alpha} ||_{l^2}
\end{array}
\label{eqTheoremEstimatesEigenvectorsH}
\end{equation}
where
$\beta_{\alpha}=\lambda_1, \lambda_2, \cdots$ for $e_{\alpha}=c_1,c_2, \cdots$; $\beta_{\alpha}=\lambda_{-1}, \lambda_{-2}, \cdots$ for $e_{\alpha}=c_{-1},c_{-2}, \cdots$; and $\beta_{\alpha}=0$ for $e_{\alpha}=d_1,d_2, \cdots$.
From Pithagoras' theorem:
\begin{equation}
3(2 \pi)^{d/2} \epsilon ||F||_{L^2 (\mathbb{R}^{2d})} > \sqrt{\sum_{\alpha \in \mathcal{A}_j^c}|\beta_{\alpha}- \mu_j|^2 |a_{\alpha}|^2} \ge m_j \sqrt{\sum_{\alpha \in \mathcal{A}_j^c} |a_{\alpha}|^2}.
\label{eqTheoremEstimatesEigenvectorsI}
\end{equation}
Let $\beta_{\alpha_0}$ be the eigenvalue of $\mathbb{F}$ such that:
\begin{equation}
|\beta_{\alpha_0}- \mu_j|=m_j.
\label{eqTheoremEstimatesEigenvectorsI.1}
\end{equation}
In other words, $\beta_{\alpha_0}$ is the eigenvalue closest to $\mu_j$. Assuming that $N$ is sufficiently large so that (\ref{eqObservation4.0},\ref{eqObservation4}) hold and also $\alpha_0 \in \left\{-N_-, \cdots, -1, 1, \cdots, N_+ \right\}$, we have from (\ref{eqTheoremEstimatesEigenvalues3}) that:
\begin{equation}
|\beta_{\alpha_0}- \mu_j| \ge |\beta_{\alpha_0}^{(N)}- \mu_j^{(N)}|- 4 ( 2 \pi)^d \epsilon ||F||_{L^2 (\mathbb{R}^{2d})}.
\label{eqTheoremEstimatesEigenvectorsI.1}
\end{equation}
And thus:
\begin{equation}
m_j \ge m_j^{(N)} -  4( 2 \pi)^d \epsilon ||F||_{L^2 (\mathbb{R}^{2d})} = M_j^{(N)} >0.
\label{eqTheoremEstimatesEigenvectorsI.2}
\end{equation}
Finally, from (\ref{eqTheoremEstimatesEigenvectorsI},\ref{eqTheoremEstimatesEigenvectorsI.2}):
\begin{equation}
\begin{array}{c}
dist\left(c_j^{(N)},\mathcal{H}_j \right) =||c_{j,\bot}^{(N)}||_{l^2}=||\sum_{\alpha \in \mathcal{A}_j^c} a_{\alpha} e_{\alpha} ||_{l^2} = \\
\\
=\sqrt{\sum_{\alpha \in \mathcal{A}_j^c} |a_{\alpha}|^2} < \frac{3 (2 \pi)^{d/2}\epsilon || F||_{L^2 (\mathbb{R}^{2d})}}{ M_j^{(N)}}
\end{array}
\label{eqTheoremEstimatesEigenvectorsJ}
\end{equation}
\end{proof}

Before we proceed, let us recall that, given some real function $F \in L^2 (\mathbb{R}^{2d})$, the Wigner function closest to $F$ in $L^2$ is $W \psi^{(0)}$, where $\psi^{(0)}$ lives in the eigenspace associated with the largest eigenvalue $\lambda_1$ of $\widehat{F}= \mathcal{W}^{-1}(F)$. If the spectrum is non-degenerate, then we choose an eigenvector $\psi_1$ associated with $\lambda_1=\mu_1$ such that
\begin{equation}
||\psi_1||_{L^2 (\mathbb{R}^d)}^2 = \lambda_1,
\label{eqRecap1}
\end{equation}
and set
\begin{equation}
W \psi^{(0)}= W \psi_1.
\label{eqRecap2}
\end{equation}
In general, we are incapable of determining the spectrum of $\widehat{F}$ or equivalently of $\mathbb{F}$ and hence we consider the truncated matrix $\mathbb{F}^{(N)}$ and its eigenvalues. Thus, instead of obtaining the exact optimizer $W\psi^{(0)}= W \psi_1$, we obtain an approximate solution $W \psi^{(0,N)} = W \psi_1^{(N)}$, where $\psi_1^{(N)}$ is an eigenvector of $\widehat{F}^{(N)}= \mathcal{W}^{-1} (F^{(N)})$ associated with the largest eigenvalue $\lambda_1^{(N)}= \mu_1^{(N)}$ of $\widehat{F}^{(N)}$ such that:
\begin{equation}
||\psi_1^{(N)}||_{L^2 (\mathbb{R}^d)}^2 = \lambda_1^{(N)}.
\label{eqRecap3}
\end{equation}
We will thus be interested in determining the distance between the exact optimizer $W \psi^{(0)}$ and the approximate solution $W \psi^{(0,N)}$. In the next theorem we obtain an estimate, if the spectrum is non-degenerate. In the degenerate case similar estimates can be obtained, but the intricacies of the calculation will depend on the multiplicity of the eigenvalues.

\begin{theorem}\label{TheoremEstimateWignerFunction}
Assume the conditions of Theorem \ref{TheoremEstimatesEigenvectors1}. Then
\begin{equation}
\begin{array}{c}
||W \psi_1 -W \psi_1^{(N)}||_{L^2 (\mathbb{R}^{2d})} \le \\
\\
\le 4 \epsilon ||F||_{L^2 (\mathbb{R}^{2d})} \left[1+ \frac{3\sqrt{\lambda_1^{(N)}} \sqrt{\lambda_1^{(N)} + 2 (2 \pi)^{d/2} \epsilon ||F||_{L^2 (\mathbb{R}^{2d})} } }{2 M_1^{(N)}} \right]
+ \frac{18 (2 \pi)^{d/2} \epsilon^2 ||F||_{L^2 (\mathbb{R}^{2d})}^2 }{(M_1^{(N)})^2}
\end{array}
\label{eqTheoremEstimateWignerFunction1}
\end{equation}
\end{theorem}

\begin{proof}
We have, using the bilinearity of the Wigner transform $W (\psi, \phi)$ and Moyal's identity:
\begin{equation}
\begin{array}{c}
||W \psi_1 -W \psi_1^{(N)}||_{L^2 (\mathbb{R}^{2d})} = ||W (\psi_1 -\psi_1^{(N)}, \psi_1) +W (\psi_1^{(N)} , \psi_1 -\psi_1^{(N)})||_{L^2 (\mathbb{R}^{2d})} \le \\
\\
\le ||W (\psi_1 -\psi_1^{(N)}, \psi_1)||_{L^2 (\mathbb{R}^{2d})}  + ||W (\psi_1^{(N)} , \psi_1 -\psi_1^{(N)})||_{L^2 (\mathbb{R}^{2d})}=\\
\\
= \frac{1}{(2 \pi )^{d/2}} \left( || \psi_1||_{L^2 (\mathbb{R}^{d})} + || \psi_1^{(N)}||_{L^2 (\mathbb{R}^{d})} \right) || \psi_1- \psi_1^{(N)}||_{L^2 (\mathbb{R}^{d})} =\\
\\
= \frac{1}{(2 \pi )^{d/2}} \left( \sqrt{\lambda_1} + \sqrt{\lambda_1^{(N)} }\right) || \psi_1- \psi_1^{(N)}||_{L^2 (\mathbb{R}^{d})} \le \\
\\
\le \frac{2 \sqrt{\lambda_1} }{(2 \pi )^{d/2}}  || \psi_1- \psi_1^{(N)}||_{L^2 (\mathbb{R}^{d})}
 \end{array}
\label{eqTheoremEstimateWignerFunction2}
\end{equation}
As before, if $\psi_1 , \psi_1^{(N)}$ are expanded in some orthonormal basis of $L^2 (\mathbb{R}^d)$ the coefficients of the expansion are eigenvectors of $\mathbb{F}$ and $\mathbb{F}^{(N)}$ with eigenvalues $\lambda_1$ and $\lambda_1^{(N)}$, respectively. If $c_1$ and $c_1^{(N)}$ are the normalized eigenvectors, the from (\ref{eqRecap1},\ref{eqRecap3}), we have:
\begin{equation}
|| \psi_1- \psi_1^{(N)}||_{L^2 (\mathbb{R}^{d})} = || \sqrt{\lambda_1} c_1 - \sqrt{\lambda_1^{(N)}} c_1^{(N)}||_{l_2}
\label{eqTheoremEstimateWignerFunction3}
\end{equation}
Since the spectrum is non-degenerate, we conclude that $c_1$ is proportional to the vector $h_1$ of $\mathcal{H}_1$ closest to $c_1^{(N)}$. By  Pithagoras' Theorem, we may thus write:
\begin{equation}
h_1= \sqrt{1- dist^2 \left(c_1^{(N)}, \mathcal{H}_1 \right)} c_1
\label{eqTheoremEstimateWignerFunction4}
\end{equation}
We chose the positive root in the previous equation, because if $c_1$ is a normalized eigenvector of $\mathbb{F}$, then so is $-c_1$. We may thus assume, without loss of generality, that $c_1$ and $h_1$ point in the same direction. From (\ref{eqTheoremEstimateWignerFunction3},\ref{eqTheoremEstimateWignerFunction4}), we have:
\begin{equation}
\begin{array}{c}
|| \psi_1- \psi_1^{(N)}||_{L^2 (\mathbb{R}^{d})} = || \sqrt{\lambda_1} c_1 - \sqrt{\lambda_1^{(N)}} h_1 + \sqrt{\lambda_1^{(N)}} h_1 - \sqrt{\lambda_1^{(N)}} c_1^{(N)}||_{l_2} \le \\
\\
\le || \sqrt{\lambda_1} c_1 - \sqrt{\lambda_1^{(N)}} h_1 ||_{l_2} + \sqrt{\lambda_1^{(N)}} || h_1 - c_1^{(N)}||_{l_2}=\\
\\
= \left| \sqrt{\lambda_1} -\sqrt{\lambda_1^{(N)}} \sqrt{1- dist^2 \left(c_1^{(N)}, \mathcal{H}_1 \right) }\right| + \sqrt{\lambda_1^{(N)}} dist \left(c_1^{(N)}, \mathcal{H}_1 \right)
\end{array}
\label{eqTheoremEstimateWignerFunction5}
\end{equation}
For the first term on the right-hand side of the previous equation, we have:
\begin{equation}
\begin{array}{c}
\left| \sqrt{\lambda_1} -\sqrt{\lambda_1^{(N)}} \sqrt{1- dist^2 \left(c_1^{(N)}, \mathcal{H}_1 \right) }\right| =\\
\\
= \frac{\lambda_1 - \lambda_1^{(N)} \left( 1- dist^2 \left(c_1^{(N)}, \mathcal{H}_1 \right)\right)}{\sqrt{\lambda_1} + \sqrt{\lambda_1^{(N)}} \sqrt{1- dist^2 \left(c_1^{(N)}, \mathcal{H}_1 \right) }} \le \frac{\lambda_1 - \lambda_1^{(N)} +   \lambda_1^{(N)} dist^2 \left(c_1^{(N)}, \mathcal{H}_1 \right)}{\sqrt{\lambda_1} }
\end{array}
\label{eqTheoremEstimateWignerFunction6}
\end{equation}
Substituting (\ref{eqTheoremEstimateWignerFunction5}) and (\ref{eqTheoremEstimateWignerFunction6}) in (\ref{eqTheoremEstimateWignerFunction2}) and taking into account the estimates (\ref{eqTheoremEstimatesEigenvalues3}) and (\ref{eqTheoremEstimatesEigenvalues17}), we obtain:
\begin{equation}
\begin{array}{c}
||W \psi_1 -W \psi_1^{(N)}||_{L^2 (\mathbb{R}^{d})} \le \\
 \\
 \le \frac{2}{(2 \pi)^{d/2}} \left[\lambda_1 - \lambda_1^{(N)}  + \lambda_1^{(N)}   dist^2 \left(c_1^{(N)}, \mathcal{H}_1 \right) +\sqrt{\lambda_1 \lambda_1^{(N)}} dist \left(c_1^{(N)}, \mathcal{H}_1 \right)\right]\le \\
 \\
 \le 4 \epsilon ||F||_{L^2 (\mathbb{R}^{2d})} \left[ 1+ \frac{9 \lambda_1^{(N)} (2 \pi)^{d/2} \epsilon ||F||_{L^2 (\mathbb{R}^{2d})}}{2 \left(M_1^{(N)}\right)^2}  + \frac{3 \sqrt{\lambda_1^{(N)}}\sqrt{\lambda_1^{(N)}+ 2 (2 \pi)^{d/2} \epsilon ||F||_{L^2 (\mathbb{R}^{2d})}}}{2 M_1^{(N)}} \right]
\end{array}
\label{eqTheoremEstimateWignerFunction7}
\end{equation}
and the result follows.
\end{proof}

\begin{remark}\label{RemarkSignalSynthesis}
Before we conclude this section, it is of interest to contrast our method with what may be called the signal synthesis approach \cite{Boudreaux,Jeong}. In the latter case one would  have to model the signal with some parameters and then calculate the Wigner distribution of the guessed at signal. Then, one would have to calculate the parameters by a least square procedure of the given distribution and modeled signal Wigner distribution.  Generally, the parameters would enter in a non-linear fashion and would have to be obtained by numerical methods. Once obtained, one would still not be sure as to whether the modeling of the signal with the guessed at parameters was optimal. In contrast our method is analytic and systematic.
\end{remark}

\section{The Wigner approximation}

In this section, we go back to the Wigner approximation for pulse propagation described in the introduction. Our aim is to prove that the Wigner approximation for a non-symplectic evolution never corresponds to a Wigner function. This justifies our search for the Wigner function closest to it. We will also illustrate our methods with a specific example.

\begin{theorem}\label{TheoremNonrepresentability}
Let
\begin{equation}
W_a \psi (x,k,t):=e^{2t \omega_I (k)} W \psi_0 (x- \nu(k)t,k),
\label{eqTheoremNonrepresentability1}
\end{equation}
be the Wigner approximation for the propagation of the initial Wigner distribution $W \psi_0 (x,k)$, with $\psi_0 \in L^2 (\mathbb{R})$ and $(x,k,t) \in \mathbb{R}^2 \times \mathbb{R}^+$. We also assume that $\omega_I, \nu \in  C (\mathbb{R})$. Then $W_a \psi_0$ is representable at some $\tau >0$ if and only if:
\begin{equation}
\omega_I(k)=0, \hspace{1 cm} \nu(k)= \nu_1 k + \nu_0, \hspace{1 cm} \nu_1,\nu_0 \in \mathbb{R}.
\label{eqTheoremNonrepresentability2}
\end{equation}
If (\ref{eqTheoremNonrepresentability2}) holds, then $W_a \psi_0 (x,k,t)$ is representable for all $t \in \left[ \right.0, + \infty \left. \right)$.
\end{theorem}

\begin{proof}
Suppose that at some instant $\tau>0$, the Wigner approximation $W_a \psi(x,k, \tau) $ is representable. Then, there exists $\phi_{\tau} \in L^2 (\mathbb{R})$ such that
\begin{equation}
W_a \psi(x,k, \tau)= W \phi_{\tau} (x,k),
\label{eqTheoremNonrepresentability3}
\end{equation}
for all $(x,k) \in \mathbb{R}^2$, and where we use the fact that Wigner functions are uniformly continuous in $\mathbb{R}^2$ \cite{Grochenig}. Let $\widetilde{\psi_0}= \mathcal{F} (\psi_0)$ and $\widetilde{\phi_{\tau}}= \mathcal{F} (\phi_{\tau})$.

We may reexpress the Wigner functions (\ref{eqHS5}) as:
\begin{equation}
W \psi_0 (x,k)= \frac{1}{2\pi} \int_{\mathbb{R}} \widetilde{\psi_0} (k - \theta /2) \overline{ \widetilde{\psi_0} (k + \theta /2)} e^{-i x \theta} d \theta,
\label{eqTheoremNonrepresentability4}
\end{equation}
and
\begin{equation}
W \phi_{\tau} (x,k)= \frac{1}{2\pi} \int_{\mathbb{R}} \widetilde{\phi_{\tau}} (k - \theta /2) \overline{ \widetilde{\phi_{\tau}} (k + \theta /2)} e^{-i x \theta} d \theta,
\label{eqTheoremNonrepresentability5}
\end{equation}
Plugging (\ref{eqTheoremNonrepresentability1},\ref{eqTheoremNonrepresentability4},\ref{eqTheoremNonrepresentability5}) into (\ref{eqTheoremNonrepresentability3}) and applying the Fourier inversion Theorem, we obtain for almost all $(\theta,k) \in \mathbb{R}^2$:
\begin{equation}
e^{2 \tau \omega_I(k)}  \widetilde{\psi_0} (k - \theta /2) \overline{ \widetilde{\psi_0} (k + \theta /2)} e^{i \theta \tau \nu(k)} =  \widetilde{\phi_{\tau}} (k - \theta /2) \overline{ \widetilde{\phi_{\tau}} (k + \theta /2)}
\label{eqTheoremNonrepresentability6}
\end{equation}
Changing variables to
\begin{equation}
p= k - \frac{\theta}{2}, \hspace{1 cm} q= k + \frac{\theta}{2}
\label{eqTheoremNonrepresentability6.1}
\end{equation}
we obtain:
\begin{equation}
e^{2 \tau \omega_I\left(\frac{p+q}{2}\right)}  \widetilde{\psi_0} (p) \overline{ \widetilde{\psi_0} (q)} e^{i \tau (q-p) \nu \left( \frac{p+q}{2}\right)} =  \widetilde{\phi_{\tau}} (p) \overline{ \widetilde{\phi_{\tau}} (q)}
\label{eqTheoremNonrepresentability6.2}
\end{equation}
for a.e. $(p,q) \in \mathbb{R}^{2}$.

Setting $p=q$, we have:
\begin{equation}
e^{2 \tau \omega_I(p)} | \widetilde{\psi_0} (p)|^2=| \widetilde{\phi_{\tau}} (p)|^2
\label{eqTheoremNonrepresentability7}
\end{equation}
for almost all $p \in \mathbb{R}$.

Let $q=k_0 \in \mathbb{R}$ be such that (\ref{eqTheoremNonrepresentability6.2}) holds for a.e. $p \in \mathbb{R}$ and $\widetilde{\phi_{\tau}} (k_0) \ne 0$. We thus have for a.e. $p \in \mathbb{R}$:
\begin{equation}
e^{2 \tau \omega_I\left(\frac{p+k_0}{2}\right) +i \tau (k_0-p)\nu\left(\frac{p+k_0}{2}\right)}  \widetilde{\psi_0} (p) \overline{ \widetilde{\psi_0}} (k_0) =  \widetilde{\phi_{\tau}} (p)\overline{ \widetilde{\phi_{\tau}} (k_0)}
\label{eqTheoremNonrepresentability9}
\end{equation}
From (\ref{eqTheoremNonrepresentability7}) and (\ref{eqTheoremNonrepresentability9}), it follows that
\begin{equation}
\widetilde{\phi_{\tau}} (p) =  \widetilde{\psi_0} (p) e^{2 \tau \omega_I(p) + i \tau (k_0-p) \nu \left(\frac{p+k_0}{2} \right) + i c_{\tau}}
\label{eqTheoremNonrepresentability10}
\end{equation}
for almost all $p \in \mathbb{R}$ and $c_{\tau} \in \mathbb{R}$ is some constant.

Upon substitution of (\ref{eqTheoremNonrepresentability10}) in (\ref{eqTheoremNonrepresentability6}), we obtain for a.e. $(\theta,k) \in \mathbb{R}^2$:
\begin{equation}
\begin{array}{c}
e^{2 \tau \omega_I(k)+ i \theta \tau \nu (k)}  \widetilde{\psi_0} (k - \theta /2) \overline{ \widetilde{\psi_0}
(k + \theta /2)}  = \widetilde{\psi_0} (k - \theta /2) \overline{ \widetilde{\psi_0} (k + \theta /2)}
\times \\
\\
\times \exp \left[ 2 \tau \omega_I \left(k - \frac{\theta}{2} \right)+ 2 \tau \omega_I \left(k + \frac{\theta}{2} \right) \right] \times \\
\\
\times \exp \left[  i \tau \left(k_0 -k + \frac{\theta}{2} \right)  \nu \left(\frac{k+k_0 - \theta/2}{2}\right) - i \tau \left(k_0 -k - \frac{\theta}{2} \right) \nu \left(\frac{k+k_0 + \theta/2}{2}\right) \right]
\end{array}
\label{eqTheoremNonrepresentability11}
\end{equation}
This is possible if and only if:
\begin{equation}
\omega_I(k)= \omega_I \left(k - \frac{\theta}{2} \right) +\omega_I \left(k + \frac{\theta}{2} \right)
\label{eqTheoremNonrepresentability12}
\end{equation}
and
\begin{equation}
\begin{array}{c}
\theta \nu (k)-  \left(k_0 -k + \frac{\theta}{2} \right)  \nu \left(\frac{k+k_0 - \theta/2}{2}\right) + \\
\\
+ \left(k_0 -k - \frac{\theta}{2} \right) \nu \left(\frac{k+k_0 + \theta/2}{2}\right) = \frac{2 \pi n_{\tau,k_0}}{\tau}
\end{array}
\label{eqTheoremNonrepresentability13}
\end{equation}
for a.e. $(\theta,k) \in \mathbb{R}^2$ and $n_{\tau ,k_0} \in \mathbb{Z}$. Since, by assumption, $\omega_I$ and $\nu$ are continuous, eqs.(\ref{eqTheoremNonrepresentability12},\ref{eqTheoremNonrepresentability13}) must in fact hold for all $(\theta,k) \in \mathbb{R}^2$. Choosing $k=k_0, ~ \theta=0$, we conclude that $n_{\tau ,k_0}=0$.  If we set $\theta=0$ in (\ref{eqTheoremNonrepresentability12}), then $\omega_I(k)=2 \omega_I(k)$, which is possible only if $\omega_I(k)$ vanishes identically. Let
\begin{equation}
k=u+v-k_0, \hspace{1 cm} \theta=2(v-u)
\label{eqTheoremNonrepresentability14}
\end{equation}
From (\ref{eqTheoremNonrepresentability13}), we obtain:
\begin{equation}
(v-u) \nu (u+v-k_0)-(k_0-u) \nu (u)+(k_0-v) \nu (v)=0.
\label{eqTheoremNonrepresentability15}
\end{equation}
The function
\begin{equation}
f(x)=(k_0-x) \nu(x)
\label{eqTheoremNonrepresentability16}
\end{equation}
is obviously continuous and satisfies:
\begin{equation}
\frac{f(u)-f(v)}{u-v} = \frac{f(u+v-k_0)}{u+v-2 k_0}.
\label{eqTheoremNonrepresentability17}
\end{equation}
Taking the limit $v \to u$ in the previous expression, we conclude that $f$ is differentiable, except possibly at $u=k_0$. If we differentiate
\begin{equation}
(u+v-2k_0) \left(f(u)-f(v) \right)=f(u+v-k_0) (u-v)
\label{eqTheoremNonrepresentability18}
\end{equation}
with respect to $u $ and with respect to v, we obtain:
\begin{equation}
\left\{
\begin{array}{l}
f(u)-f(v)+ (u+v-2 k_0) f^{\prime} (u)=f^{\prime} (u+v-k_0) (u-v) +f(u+v-k_0)\\
f(u)-f(v)- (u+v-2 k_0) f^{\prime} (v)=f^{\prime} (u+v-k_0) (u-v) -f(u+v-k_0)
\end{array}
\right.
\label{eqTheoremNonrepresentability19}
\end{equation}
If we subtract the two equations, we obtain:
\begin{equation}
(u+v-2 k_0) \left(f^{\prime} (u)+ f^{\prime} (v) \right)=2 f(u+v-k_0),
\label{eqTheoremNonrepresentability20}
\end{equation}
and setting $v=u$:
\begin{equation}
2(u-k_0) f^{\prime} (u)= f(2u-k_0), \hspace{1 cm} \mbox{for all } u \in \mathbb{R} \backslash \left\{k_0 \right\}
\label{eqTheoremNonrepresentability21}
\end{equation}
This means that $f$ and hence $\nu$ is twice differentiable except possibly at $k_0$. If we go back to (\ref{eqTheoremNonrepresentability15}) and differentiate first with respect to $u$ and then with respect to $v$, we conclude that:
\begin{equation}
(v-u) \nu^{\prime \prime} (u+v-k_0)=0.
\label{eqTheoremNonrepresentability22}
\end{equation}
Thus for $v=k_0$ and $u \ne k_0$, $\nu^{\prime \prime} (u)=0$ and (\ref{eqTheoremNonrepresentability2}) follows.

Finally, if (\ref{eqTheoremNonrepresentability2}) holds, then the Wigner approximation amounts at all times to an affine linear symplectic transformation. It is well-known that under these circumstances it must be a representable function \cite{Folland,Grochenig}.
\end{proof}

As a concrete example, we consider the standard centered Gaussian for $d=1$ as the initial Wigner function
\begin{equation}
W \psi_0 (z)= \frac{1}{\pi} e^{- |z|^2}
\label{eqTheoremNonrepresentability23}
\end{equation}
and choose the dispersion relation
\begin{equation}
\omega (k)= \frac{k^3}{3}.
\label{eqTheoremNonrepresentability24}
\end{equation}
The Wigner approximation (\ref{eqIntroduction6}) is then given by:
\begin{equation}
W_a \psi (x,k,t)= W \psi_0 \left(x- k^2 t , k \right) = \frac{1}{\pi} e^{-\left(x- k^2t \right)^2 - k^2}.
\label{eqTheoremNonrepresentability25}
\end{equation}
We thus want to obtain the Wigner function closest to
\begin{equation}
\begin{array}{c}
F(x,k,t)= W_a \psi (x,k,t) = \\
\\
=\frac{1}{\pi} e^{-x^2 -k^2} \left( 1+2 x k^2 t+ k^4(2x^2-1)t^2 \right) + \mathcal{O} (t^3).
\end{array}
\label{eqTheoremNonrepresentability25.1}
\end{equation}
The expansion coefficients (\ref{eq13}) read:
\begin{equation}
\begin{array}{c}
f_{n,m} (t)= 2 \pi <W_a \psi, W (e_n,e_m)>_{L^2 (\mathbb{R}^2)} = \\
\\
=2 \int_{\mathbb{R}} \int_{\mathbb{R}} e^{-\left(x- k^2t \right)^2 - k^2} \overline{W (e_n, e_m) (x,k)} d x d k
\end{array}
\label{eqTheoremNonrepresentability26}
\end{equation}
The integral in the previous expression is uniformly convergent for all $t \in \mathbb{R}$
\begin{equation}
|f_{n,m} (t)| \le 2 \int_{\mathbb{R}} \int_{\mathbb{R}}  \left|W (e_n, e_m) (x,k) \right| d x d k = 2 ||W(e_n,e_m)||_{L^1 (\mathbb{R}^2)} < \infty
\label{eqTheoremNonrepresentability27}
\end{equation}
We conclude that $f_{n,m} \in C^{\infty} (\mathbb{R})$. We calculate some coefficients to order $\mathcal{O} (t^2)$.
\begin{equation}
f_{n,m} (t)= 2 \int_{\mathbb{R}} \int_{\mathbb{R}} e^{-x^2 - k^2} \left(1+ 2 xk^2 t + k^4 (2 x^2 -1) t^2  \right) \overline{W (e_n, e_m) (x,k)} d x d k + \mathcal{O} (t^3)
\label{eqTheoremNonrepresentability28}
\end{equation}
Using the Hermite basis (\ref{eqHermiteFunctions1}-\ref{eqHermiteFunctions4}), we obtain
\begin{equation}
\begin{array}{l}
f_{0,0} (t)=1- \frac{3}{32} t^2 + \mathcal{O} (t^3)\\
\\
f_{0,1} (t)= f_{1,0} (t)= \frac{\sqrt{2}}{8} t + \mathcal{O} (t^3)\\
\\
f_{1,1} (t)=- \frac{3}{32} t^2 + \mathcal{O} (t^3)
\end{array}
\label{eqTheoremNonrepresentability29}
\end{equation}
Thus
\begin{equation}
\begin{array}{c}
F^{(2)} (x,k,t) = f_{0,0} (t) W(e_0,e_0) (x,k) + f_{0,1} (t) W(e_0,e_1) (x,k) +\\
\\
+ f_{1,0} (t) W(e_1,e_0) (x,k) +f_{1,1} (t) W(e_1,e_1) (x,k) =\\
\\
= \frac{1}{\pi} e^{-x^2-k^2} \left(1+ \frac{tx}{2} - \frac{3 t^2}{16} (x^2+k^2) \right) + \mathcal{O} (t^3)
\end{array}
\label{eqTheoremNonrepresentability29.1}
\end{equation}
and
\begin{equation}
\mathbb{F}^{(2)} = \left(
\begin{array}{c c}
1- \frac{3}{32}  t^2 & \frac{\sqrt{2}}{8} t\\
& \\
 \frac{\sqrt{2}}{8} t & - \frac{3}{32} t^2
\end{array}
\right) + \mathcal{O} (t^3)
\label{eqTheoremNonrepresentability30}
\end{equation}
The eigenvalues of $\mathbb{F}^{(2)}$ are:
\begin{equation}
\begin{array}{l}
\lambda_1^{(2)}= \frac{1}{2} \left( 1- \frac{3}{16} t^2 + \sqrt{1+ \frac{t^2}{8}}\right) =  1- \frac{t^2}{16}  + \mathcal{O} (t^3)\\
\\
\lambda_{-1}^{(2)}= \frac{1}{2} \left( 1- \frac{3}{16} t^2 - \sqrt{1+ \frac{t^2}{8}}\right) =  - \frac{t^2}{4}  + \mathcal{O} (t^3)
\end{array}
\label{eqTheoremNonrepresentability31}
\end{equation}
Notice that $\lambda_{-1}^{(2)} <0$ for $t>0$. This is in agreement with the result of Theorem \ref{TheoremNonrepresentability}. The eigenvectors associated with $\lambda_1^{(2)} $ read
\begin{equation}
c= K_t \left( \frac{\sqrt{2}}{4} t , \sqrt{1+ \frac{t^2}{8}} -1\right)
\label{eqTheoremNonrepresentability32}
\end{equation}
where $K_t$ is an arbitrary (non-zero) complex function of time only.

The minimizing wave function is given by:
\begin{equation}
\psi_1^{(2)} (x)= K_t \left(\frac{\sqrt{2}}{4} t e_0 (x) +  \left(\sqrt{1+ \frac{t^2}{8}} -1\right) e_1(x) \right)
\label{eqTheoremNonrepresentability33}
\end{equation}
If we impose (\ref{eq9}), we obtain:
\begin{equation}
||\psi_1^{(2)} ||_{L^2 (\mathbb{R})}^2 = \lambda_1^{(2)} \Leftrightarrow K_t= \frac{1}{2} \sqrt{\frac{1- \frac{3}{16} t^2 + \sqrt{1+ \frac{t^2}{8}}}{1+ \frac{1}{8} t^2 - \sqrt{1+ \frac{t^2}{8}}}}
\label{eqTheoremNonrepresentability34}
\end{equation}
Consequently:
\begin{equation}
\psi_1^{(2)} (x)= \left(1- \frac{t^2}{32} \right) e_0(x) + \frac{t}{4 \sqrt{2}} e_1 (x) + \mathcal{O} (t^3)
\label{eqTheoremNonrepresentability35}
\end{equation}
Thus, to this order the Wigner function closest to the Wigner approximation $F=W_a \psi$ (\ref{eqTheoremNonrepresentability25}) is
\begin{equation}
\begin{array}{c}
W \psi^{(0,2)} (x,k,t) =W \psi_1^{(2)} (x,k,t) = \\
 \\
 = \frac{1}{\pi} e^{- |z|^2} \left\{ 1+ \frac{tx}{2}  + \frac{t^2}{32} (2 |z|^2 -3) \right\} + \mathcal{O} (t^3)
\end{array}
\label{eqTheoremNonrepresentability36}
\end{equation}
Let us estimate the error of the truncation. From (\ref{eq16},\ref{eqTheoremNonrepresentability25.1},\ref{eqTheoremNonrepresentability29.1}), we obtain
\begin{equation}
\begin{array}{c}
\epsilon^2 ||F||_{L^2 (\mathbb{R}^2)}^2 > ||F- F^{(2)}||_{L^2 (\mathbb{R}^2)}^2 =\\
 \\
 =\left(\frac{t}{2 \pi} \right)^2 \int_{\mathbb{R}} \int_{\mathbb{R}} e^{-2 x^2-2k^2} x^2 (4k^2-1)^2 dx d k +  \mathcal{O} (t^3) = \frac{t^2}{16 \pi} +  \mathcal{O} (t^3)
 \end{array}
\label{eqTheoremNonrepresentability39}
\end{equation}
We conclude that
\begin{equation}
|\lambda_1 - \lambda_1^{(2)}| < \frac{t}{\sqrt{2}} +  \mathcal{O} (t^2)
\label{eqTheoremNonrepresentability40}
\end{equation}
We also have
\begin{equation}
M_1 ^{(2)} = \frac{1}{2} \left(1- 4 \sqrt{2} t - \frac{3}{16} t^2 \right)+  \mathcal{O} (t^3)
\label{eqTheoremNonrepresentability40}
\end{equation}
This approximation makes sense for sufficiently small $t$ in order that $M_1 ^{(2)}  > 0$ (cf.(\ref{eqObservation5}))
Moreover, from (\ref{eqTheoremEstimatesEigenvalues17}) we conclude that
\begin{equation}
dist \left(c_1^{(2)}, \mathcal{H}_1 \right) < \frac{3}{\sqrt{2}} t (1+ 2 \sqrt{\pi}t ) +  \mathcal{O} (t^3)
\label{eqTheoremNonrepresentability41}
\end{equation}
Finally, estimate (\ref{eqTheoremEstimateWignerFunction1}) yields:
\begin{equation}
||W \psi_1-W \psi_1^{(2)}||_{L^2 (\mathbb{R}^{2d})} \le \frac{4t}{\sqrt{\pi}}\left(1+ \frac{9t}{\sqrt{2}} \right) + \mathcal{O} (t^3)
\label{eqTheoremNonrepresentability41A}
\end{equation}

Thus if we choose for instance $t<0.01$, then $M_1 ^{(2)}\simeq 0.478$ is positive and we have
\begin{equation}
|\lambda_1- \lambda_1^{(2)}| \lesssim 0.008, \hspace{1 cm} dist \left(c_1^{(2)}, \mathcal{H}_1 \right) \lesssim 0.022,
\label{eqTheoremNonrepresentability42}
\end{equation}
and
\begin{equation}
 ||W \psi_1-W \psi_1^{(2)}||_{L^2 (\mathbb{R}^{2d})} \lesssim 0.024
\label{eqTheoremNonrepresentability43}
\end{equation}

\section{Schatten-class operators}

As another application of our results, we can estimate the eigenvalues and the norms of certain Schatten-class operators. Let us briefly recall the definition of Schatten-von Neumann operators \cite{Exner}. Let $p \in \left[ \right.1, \infty \left[ \right.$. Given some operator $\widehat{A}$ acting on a separable Hilbert space $\mathcal{H}$, we denote by:
\begin{equation}
|\widehat{A}| = (\widehat{A}^{\ast} \widehat{A})^{1/2}
\label{eqScahtten1}
\end{equation}
the positive root of $\widehat{A}^{\ast} \widehat{A}$, where $\widehat{A}^{\ast}$ is the adjoint of $\widehat{A}$. Its $p$-th Schatten norm is given by:
\begin{equation}
|| \widehat{A}||_{S_p (\mathcal{H})} =  \left(Tr |\widehat{A}|^p \right)^{1/p}.
\label{eqScahtten2}
\end{equation}
The trace of an operator $\widehat{B}$ is given by:
\begin{equation}
Tr (\widehat{B})= \sum_n < \widehat{B} e_n, e_n>_{\mathcal{H}},
\label{eqScahtten3}
\end{equation}
for some orthonormal basis $\left\{e_n \right\}_n$. If it is finite, then the result is independent of the orthonormal basis chosen.

An operator $\widehat{A}$ belongs to the $p$-th Schatten class $S_p (\mathcal{H})$ if its $p$-th Schatten norm (\ref{eqScahtten2}) is finite. Schatten class operators are compact.
Particular cases are the trace-class operators $(p=1)$ and the Hilbert-Schmidt operators $(p=2)$.

If $\widehat{A} \in S_p (\mathcal{H})$ is self-adjoint, then it admits a decomposition of the form (\ref{eqHSOp4})-(\ref{eqHSOp9}). We can thus write its $p$-the Schatten-norm as:
\begin{equation}
||\widehat{A} ||_{S_p (\mathcal{H})} = \left(  \sum_{j \in\mathbb{U}_+} \mu_j^p + \sum_{j \in\mathbb{U}_-} |\mu_{-j}|^p \right)^{1/p} = \left(  \sum_{\alpha \in \mathbb{U}} |\mu_{\alpha}|^p \right)^{1/p} .
\label{eqScahtten4}
\end{equation}
 We have the following continuous embedding:
\begin{equation}
||\widehat{A}||_{S_p (\mathcal{H})} \le ||\widehat{A}||_{S_q (\mathcal{H})}  \Rightarrow S_q (\mathcal{H}) \subset S_p (\mathcal{H}), \hspace{1 cm} 1 \le q \le p < \infty
\label{eqScahtten5}
\end{equation}
From our previous results we obtain the following two propositions.

Let $\widehat{F} \in S_p \left( L^2 (\mathbb{R}^d) \right)$ for some $p \in \left[1, 2 \right]$, with $\widehat{F}$ self-adjoint. From (\ref{eqScahtten5}) $\widehat{F}$ is a Hilbert-Schmidt operator with Weyl symbol $F= \mathcal{W} (\widehat{F})$ and hence it admits the matrix representation $\mathbb{F}$ as before with respect to some orthonormal basis. With the assumption (\ref{eq16}), we have a truncated matrix $\mathbb{F}^{(N)}$ with the associated eigenvalues $\left\{\mu_j^{(N)} \right\}_j$. From Theorem \ref{TheoremEstimatesEigenvalues} it follows that:
\begin{proposition}\label{PropositionSchatten1}
Let $\widehat{F} \in S_p \left( L^2 (\mathbb{R}^d) \right)$ for some $p \in \left[1, 2 \right]$. Under the assumption (\ref{eq16}), we have:
\begin{equation}
\mu_j^{(N)} \le \mu_j \le max \left\{ \mu_j^{(N)} +(2 \pi)^{d/2} \epsilon ||F||_{L^2 (\mathbb{R}^{2d})}, 2 (2 \pi)^{d/2} \epsilon ||F||_{L^2 (\mathbb{R}^{2d})} \right\},
\label{eqScahtten6}
\end{equation}
for $j=1,2, \cdots, N_+$, and
\begin{equation}
min \left\{ - \mu_{-j}^{(N)} - (2 \pi)^{d/2} \epsilon ||F||_{L^2 (\mathbb{R}^{2d})}, - 2 (2 \pi)^{d/2} \epsilon ||F||_{L^2 (\mathbb{R}^{2d})}  \right\} \le \mu_{-j} \le \mu_{-j}^{(N)},
\label{eqScahtten7}
\end{equation}
for $j=1,2, \cdots, N_-$. Consequently,
\begin{equation}
|\mu_j^{(N)} - \mu_j| < 2 (2 \pi)^{d/2} \epsilon ||F||_{L^2 (\mathbb{R}^{2d})},
\label{eqScahtten8}
\end{equation}
for all $j=-N_{-}, \cdots,-1,1, \cdots, N_+$.
\end{proposition}

Schatten norms are, in general, very difficult to compute. The exception is the Hilbert-Schmidt norm, because it can be evaluated through the $L^2$ norm of the Weyl symbol (cf. \ref{eqMoyalIdentity1},\ref{eq3}):
\begin{equation}
||\widehat{F}||_{S_2 \left(L^2 (\mathbb{R}^d)\right)} = (2 \pi)^{d/2} || F||_{L^2 (\mathbb{R}^{2d})}.
\label{eqScahtten9}
\end{equation}
But for all the other Schatten norms, there is no such simple formula and one is forced to determine the complete spectrum of $ \widehat{F}$ to compute (\ref{eqScahtten4}). Our results permit to approximate some Schatten norms. We consider this problem again from another perspective elsewhere \cite{Ben}.

\begin{proposition}\label{PropositionSchatten1}
Let $\widehat{F} \in S_p \left( L^2 (\mathbb{R}^d) \right)$ for some $p \in \left[1, 2 \right]$, with $\widehat{F}$ self-adjoint. Under the assumption (\ref{eq16}), we have for any $q \in \left[ \right. 2, \infty \left[ \right.$:
\begin{equation}
\left| ~ ||\widehat{F} ||_{S_q \left(L^2 (\mathbb{R}^d)\right)} - \left(\sum_{j=1}^{N_+} \left(\mu_j^{(N)} \right)^q + \sum_{j=1}^{N_-} \left|\mu_{-j}^{(N)} \right|^q \right)^{1/q} \right| \le (2 \pi)^{d/2} \epsilon || F||_{L^2 (\mathbb{R}^{2d})}.
\label{eqScahtten10}
\end{equation}

\begin{proof}
Under the assumption  (\ref{eq16}), we have:
\begin{equation}
\begin{array}{c}
||\widehat{F} ||_{S_q \left(L^2 (\mathbb{R}^d)\right)} = ||\widehat{F}- \widehat{F}^{(N)} + \widehat{F}^{(N)} ||_{S_q \left(L^2 (\mathbb{R}^d)\right)} \le \\
\\
\le ||\widehat{F}^{(N)} ||_{S_q \left(L^2 (\mathbb{R}^d)\right)} + ||\widehat{F}- \widehat{F}^{(N)} ||_{S_q \left(L^2 (\mathbb{R}^d)\right)} \le \\
\\
\le ||\widehat{F}^{(N)} ||_{S_q \left(L^2 (\mathbb{R}^d)\right)} + ||\widehat{F}- \widehat{F}^{(N)} ||_{S_2 \left(L^2 (\mathbb{R}^d)\right)} \le \\
\\
\le ||\widehat{F}^{(N)}||_{S_q \left(L^2 (\mathbb{R}^d)\right)} +  (2 \pi)^{d/2} \epsilon || F||_{L^2 (\mathbb{R}^{2d})},
\end{array}
\label{eqScahtten11}
\end{equation}
where we used (\ref{eq16},\ref{eqScahtten5},\ref{eqScahtten9}).

From (\ref{eqScahtten4}), we have:
\begin{equation}
||\widehat{F}^{(N)}||_{S_q \left(L^2 (\mathbb{R}^d)\right)}= \left(\sum_{j=1}^{N_+} \left(\mu_j^{(N)} \right)^q + \sum_{j=1}^{N_-} \left|\mu_{-j}^{(N)} \right|^q \right)^{1/q}.
\label{eqScahtten12}
\end{equation}
Finally from the monotonicity of the eigenvalues (\ref{eqPropositionSequence1},\ref{eqPropositionSequence2}), we have that
\begin{equation}
\left|~||\widehat{F} ||_{S_q \left(L^2 (\mathbb{R}^d)\right)}  - ||\widehat{F}^{(N)}||_{S_q \left(L^2 (\mathbb{R}^d)\right)}~ \right| = ||\widehat{F} ||_{S_q \left(L^2 (\mathbb{R}^d)\right)}  - ||\widehat{F}^{(N)}||_{S_q \left(L^2 (\mathbb{R}^d)\right)},
\label{eqScahtten12}
\end{equation}
and the result follows.
\end{proof}

\end{proposition}

\section{Conclusions and outlook}

Let us briefly recapitulate our results. Let $F \in L^2 (\mathbb{R}^{2d})$ be some real non-representable function. By this, we mean that there exists no $\psi \in L^2 (\mathbb{R}^d)$ such that $F=W \psi$. We then look for the Wigner function $W \psi_0$ which is closest to $F$ in the $L^2$ norm.

We solved this problem exactly in the case where $F$ is a one-dimensional radial function. For the general case, we used a truncated version $\mathbb{F}^{(N)}=\left\{F_{n,m}\right\}_{1 \leq n,m \leq N}$ of the complete expansion coefficients $\mathbb{F}=\left\{F_{n,m}\right\}_{n,m \in \mathbb{N}}$ of the function $F$ in a given orthonormal basis of Wigner functions $\left\{W(e_n,e_m) \right\}_{n,m}$. By resorting to the Courant-Fischer min-max theorem, we obtained precise estimates for the errors of the approximate eigenvalues and eigenvectors. We proved that the function $F=W_a \psi$ obtained by the Wigner approximation method developed in \cite{Cohen2,Loughlin,Loughlin1,Loughlin2} is never a Wigner function. We then used our methods to determine approximately the Wigner function closest to $W_a \psi$. Finally, we have shown that certain Schatten norms of self-adjoint Schatten class operators can be evaluated to any precision with our methods.

In a future work, we wish to study other quasi-distributions. In the previous sections we have used only the Wigner distribution. However, there
are an infinite number of other phase space distributions \cite{Cohen1,Cohen3,Cohen4} and a number of
questions arise when the considerations of the previous sections
are applied to other distributions. We briefly discuss the general class of quasi-distributions. For simplicity we will consider the one-dimensional case.

One can characterize the distributions by way of the kernel method. All
bilinear distributions are given by
\begin{equation}
\begin{array}{c}
W^{\Phi} \psi (x,k,t)={\frac{1}{{4\pi^{2}}}}\iiint\overline{{\psi} (\,x^{\prime
}-\,{{\tfrac{1}{2}}}\tau,t)} \,{\psi}(\,x^{\prime}+\,{{\tfrac{1}{2}}}\tau
,t) \\
\\
\Phi(\theta,\tau)e^{-i\theta x-i\tau{k+}i\theta\,x^{\prime}}\,d\theta
\,d\tau\,d\,x^{\prime}%
\end{array}
\label{eqquasidistributions1}
\end{equation}
where $\Phi(\theta,\tau)$ is called the kernel and characterizes the
particular distribution. For the Wigner distribution, $\Phi(\theta,\tau)=1$. Here is how
one can understand the previous expression. We assume that $\psi \in \mathcal{S} (\mathbb{R})$ and hence $W \psi \in \mathcal{S} (\mathbb{R}^2)$. By inverting the partial Fourier transform with respect to the second variable in (\ref{eqIntroduction1}), we obtain:
\begin{equation}
\overline{{\psi} (\,x^{\prime
}-\,{{\tfrac{1}{2}}}\tau,t)}  \,{\psi}(\,x^{\prime}+\,{{\tfrac{1}{2}}}\tau
,t)= \int_{\mathbb{R}} W \psi (x^{\prime},k^{\prime},t) e^{i k^{\prime} \tau} d k^{\prime}
\label{eqquasidistributions2}
\end{equation}
Let $\Phi \in \mathcal{S}^{\prime} (\mathbb{R}^2)$ and let $\widetilde{\Phi}$ be its Fourier transform:
\begin{equation}
\widetilde{\Phi} (x,k) = \left(\mathcal{F} \Phi \right) (x,k)= \frac{1}{2 \pi} \iint \Phi (\theta, \tau) e^{-i \theta x - i \tau k} d \theta d \tau,
\label{eqquasidistributions3}
\end{equation}
which should be understood in the usual distributional sense:
\begin{equation}
<\mathcal{F} \Phi , F>=<\Phi,\mathcal{F} F>,
\label{eqquasidistributions4}
\end{equation}
for all $F \in \mathcal{S} (\mathbb{R}^2)$ and where $< \cdot, \cdot>$ denotes the distributional bracket.

If we use (\ref{eqquasidistributions3}) and plug (\ref{eqquasidistributions2}) into (\ref{eqquasidistributions1}), we obtain:
\begin{equation}
W^{\Phi} \psi (x,k,t)= \frac{1}{2 \pi} \left(\widetilde{\Phi} \star W \psi \right) (x,k,t)
\label{eqquasidistributions5}
\end{equation}
where $\star$ denotes the convolution.

Some explicit relations between distributions are as follows. Two
distributions {\normalsize $W^{\Phi_{1}} \psi$ and $W^{\Phi_{2}} \psi$ }characterized by the kernels
{\normalsize $\Phi_{1}$ and $\Phi_{2}$ are related by
\begin{equation}
\begin{array}{c}
W^{\Phi_{2}} \psi(x,k,t)\,=\,\frac{1}{4\pi^{2}}%
{\displaystyle\iiiint}
\frac{\Phi_{2}(\theta,\tau)}{\Phi_{1}(\theta,\tau)}\,W^{\Phi_{1}} \psi (x^{\prime
},k^{\prime},t)\, \\
\\
e^{i\theta(x^{\prime}-x)\,+\,i\tau(k^{\prime}-k)}%
\,d\theta\,d\tau\,dx^{\prime}\,dk^{\prime} \label{c2}%
\end{array}
\end{equation}
} Eq. (\ref{c2})\ can be expressed in the form of a pseudo-differential operator,%
\begin{equation}
W^{\Phi_{2}} \psi({x},{k,t})=\,\Phi_{2}\left(  i\frac{\partial{}}{\partial{x}%
},i\frac{\partial{}}{\partial{k}}\right)  \Phi_{1}^{-1} \left(  i\frac{\partial{}%
}{\partial{x}},i\frac{\partial{}}{\partial{k}}\right)  W^{\Phi_{1}} \psi({x},{k,t})
\label{eqquasidistributions6}
\end{equation}

Just as for the Wigner distribution, one can calculate expectation values by resorting to other quasi-distributions. Let $\widehat{A}$ be some Weyl-operator with Weyl-symbol $a \in \mathcal{S} (\mathbb{R}^2)$. Given $\psi \in \mathcal{S} (\mathbb{R})$, we have:
\begin{equation}
<\widehat{A} \psi, \psi >_{L^2 (\mathbb{R})} = \int_{\mathbb{R}} \int_{\mathbb{R}} a^{\Phi} (x,k) W^{\Phi} (x,k,t) dx dk.
\label{eqquasidistributions7}
\end{equation}
Here $a^{\Phi}$ is the symbol of $\widehat{A}$ associated with the kernel $\Phi$. It is related with the Weyl symbol $\Phi$ according to:
\begin{equation}
a= \frac{1}{2 \pi} \widetilde{\Phi} \star a^{\Phi}.
\label{eqquasidistributions8}
\end{equation}

One can then ask whether one can apply the same ideas as described in the introduction to other quasi-distributions and obtain analogous approximations. This has been
partially answered and we describe two such cases. The first is the
Margenau-Hill distribution, the kernel for which is%
\begin{equation}
\Phi_{MH}(\theta,\tau)=e^{-i\theta\tau/2}%
\label{eqquasidistributions9}
\end{equation}
which results in the distribution%
\begin{equation}
W^{MH} \psi ({x,k,t})={\frac{1}{\sqrt{{2\pi}}}}\overline{\psi({x},t)} e^{i{kx}}%
\left(\mathcal{F} \psi \right)({k},t)
\label{eqquasidistributions10}
\end{equation}
where $\mathcal{F} \psi$ is the Fourier transform of $\psi$.

Proceeding analogously as with the approximation for the Wigner distribution,
we obtain the Margenau-Hill approximation,
\begin{equation}
W^{MH}  \psi (x,{k},t)\approx e^{2t\omega_{I}(k)}W^{MH}  \psi_0 (x-\nu(k)t,k) \label{MH}%
\end{equation}

For the spectrogram, the kernel is%
\begin{equation}
\Phi_{SP}(\theta,\tau)=\!\!\int_{\mathbb{R}} \overline{w({x}+{{\tfrac{\tau}{2}}})}%
\,e^{-i\theta{x}}w({x}-{{\tfrac{\tau}{2}}})\,d{x}%
\end{equation}
where $w({x}) \in \mathcal{S} (\mathbb{R})$ is the window function. The distribution is
\begin{equation}
W^{SP} \psi ({x},{k})\,=\,\left\vert \,\frac{1}{\sqrt{2\pi}}\int_{\mathbb{R}} \,e^{-i kx^{\prime}}\psi({x}^{\prime},t)\, \overline{w({x}^{\prime}-{x})}d{x}^{\prime
}\,\right\vert ^{2}%
\label{eqSpectrogram1}
\end{equation}
For a particular window $w$, this can be seen as the modulus squared of the Fourier-Bros-Iagolnitzer (FBI) transform \cite{Folland}.

Analogous to the Wigner distribution, the approximation works out as
\begin{equation}
W^{SP} \psi (x,{k},t)\approx e^{2t\omega_{I}(k)}W^{SP} \psi_0 (x-\nu(k)t,k)\label{SP}%
\end{equation}

Comparing the Wigner approximation, Eq. (\ref{eqIntroduction6})\ with the MH approximation
(\ref{MH}) and the spectrogram, Eq. (\ref{SP})\ we see that they are of the
same functional form. This gives rise to various questions
that we discuss and are currently being studied.

\begin{itemize}
\item While the approximations are of the same functional form, the accuracy of
the approximations is not necessarily equivalent. One can ask, which is
closest to the exact corresponding distribution? Second, which produces a more
accurate wave function by whatever method one can use to invert the
distribution and obtain an approximate wave function?

\bigskip

\item It is probably the case that none of the approximations are representable.
Can one define approximate representability and see which distribution is most representable?

\bigskip

\item What are the next (higher-order) approximations and are they the same for
the different distributions?

\bigskip

\item Related to the previous issues, if one can find a series approximation,
will the successive approximations be more and more representable?

\bigskip

\item Another method of approximation is the differential equation approach. Does
that approach give the same approximations.

\item The $L^2$ norm seems natural in this setting because Wigner functions belong to $L^2 (\mathbb{R}^{2d})$ and Moyal's identity leads to natural orthogonality relations. Moreover it tends to be pervasive in physical applications. However, in statistical estimation the fundamental measure is
the $L^1$ norm, since it allows us to control the MSE error in the estimations
(this is detailed in \cite{Abreu3} for stationary signals). In a future work we will try to investigate whether it is also possible to find the Wigner function closest to a given function in phase space with respect to the $L^1$ norm.

\end{itemize}

\section*{Acknowledgements}

The work of N.C. Dias and J.N. Prata is supported by the Portuguese Science Foundation (FCT) grant PTDC/MAT-CAL/4334/2014. The authors would like to thank Franz Luef for drawing their attention to references \cite{Balazs1,Balazs2}.

\pagebreak

**********************************************************************************************************************************************************************************************************

\textbf{Author's addresses:}

\begin{itemize}
\item \textbf{J.S. Ben-Benjamin:} Institute for Quantum Science and Engineering, Texas AM University, College Station USA

\item \textbf{L. Cohen:} Department of Physics, Hunter College of the City University of New York, 695 Park Ave. New York, NY 10021 USA

\item \textbf{P. Loughlin:} Department of Bioengineering, University of Pittsburgh, Pittsburgh, PA 15261, USA

\item \textbf{N.C. Dias and J.N. Prata: }Grupo de F\'{\i}sica
Matem\'{a}tica, Departamento de Matem\'atica, Universidade de Lisboa, Av. Campo Grande, Edif\'{\i}cio C6, 1749-016
Lisboa, Portugal, and Escola Superior N\'autica Infante D. Henrique, Av. Engenheiro Bonneville Franco, 2770-058 Pa\c{c}o de Arcos, Portugal

\end{itemize}

**********************************************************************************************************************************************************************************************************


\begin{thebibliography}{9}                                                                                                %

\bibitem{Abreu1} L. D. Abreu, K. Gr\"ochenig. Banach Gabor frames with Hermite functions:
polyanalytic spaces from the Heisenberg group. Appl. Anal. 91 (2012) 1981-1997.

\bibitem{Abreu2} L. D. Abreu,  J. M. Pereira. Measures of localization and quantitative
Nyquist densities. Appl. Comp. Harm. Anal. 38 (2015) 524-534.

\bibitem{Abreu3} L. D. Abreu,  J. L. Romero. MSE bounds for multitaper spectral estimation
and off-grid compressive sensing. IEEE Trans. Inform. Theory 63 (12) (2017) 7770-7776.

\bibitem{Abreu4} L. D. Abreu, H. G. Feichtinger. Function spaces of polyanalytic functions. Harmonic and Complex Analysis and its Applications. Trends in Mathematics, Springer, pp. 1-38 (2014).

\bibitem{Balazs1} P. Balazs. Basic definition and properties of Bessel multipliers. J. Math. Anal. and Appl. 325 (2007) 571-585.

\bibitem{Balazs2} P. Balazs, D. Bayer and A. Rahimi. Multipliers for continuous frames in Hilbert spaces. J. Phys. A: Math. Theor. 45 (2012) 244023 (20pp).

\bibitem{Bayram} M. Bayram and R. G. Baraniuk. Multiple window time-varying spectrum
estimation. Nonlinear and Nonstationary Signal Processing, Cambridge,
U.K.: Cambridge Univ. Press, 2000, 292316.


\bibitem{Ben} J.S. Ben-Benjamin, N.C. Dias, L. Cohen, P. Loughlin, J.N. Prata. On the topology of Wigner functions. Submitted.

\bibitem{Boudreaux} G.F. Boudreaux-Bartels and T.W. Parks. Time-varying and signal estimation using Wigner distribution synthesis techniques. IEEE Tans. Accoust., Speech, and Signal Process. ASSP-34(3) (1986) 442-451.

\bibitem{Exner} J. Blank, P. Exner, M. Havl\'{\i}$\check{\mbox{c}}$ek. Hilbert space operators in quantum physics. 2nd Edition (2008), Springer.

\bibitem{Bracken} A.J. Bracken, H.D. Doebner, J.G. Wood. Bounds on integrals of the Wigner function. Phys. Rev. Lett. 83 (1999) 3758.

\bibitem{Cohen1} L. Cohen. The Weyl operator and its generalization. Pseudo-Differential Operators. Theory and Applications, Vol.9 Birkh\"auser (2013).

\bibitem{Cohen3} L. Cohen. Time-frequency distributions - A review. Proc. IEEE 77 (1989) 941 - 981.

\bibitem{Cohen4} L. Cohen. Generalized phase-space distribution functions. J. Math. Phys. 7 (1966) 781.

\bibitem{Cohen2} L. Cohen, P. Loughlin, G. Okopal. Exact and approximate moments of a propagating pulse. J. Mod. Optics 55 (2008) 3349-3358.

\bibitem{Daubechies} I. Daubechies. Time-frequency localization operators: a geometric phase space approach.
IEEE Trans. Inform. Theory, 34(4):605–612, 1988.

\bibitem {Dias1}N.C. Dias and J.N. Prata. Admissible states in quantum phase space. Ann. Phys. 313 (2004) 110--146.

\bibitem{Flandrin} P. Flandrin. Maximum signal energy concentration in a time-frequency domain. Proc. IEEE Int. Conf. Accoustics, Speech, Signal Processing (ICASSP `88) 4 (1988) 2176-2179, New York, USA.

\bibitem{Folland} G.B. Folland. Harmonic analysis in phase space. Annals of Mathematics Studies, Vol. 122, Princeton University Press, Princeton, NJ (1989).

\bibitem{Gosson} M. de Gosson. Symplectic geometry and quantum mechanics. Birkh\"auser, Basel (2006).

\bibitem{Griffin} D.W. Griffin and J.S. Lim. Signal estimation from modified short-time Fourier transform.
IEEE Trans. Acoust., Speech, and Signal Process. ASSP 32(2):236–243, 1984.

\bibitem{Grochenig}K. Gr\"{o}chenig, \textit{Foundations of Time-Frequency
Analysis}, Birkh\"{a}user, Boston, (2000).

\bibitem{Haimi} A. Haimi, H. Hedenmalm. The polyanalytic Ginibre ensembles. J.
Stat. Phys. 153 (1) (2013) 10-47.

\bibitem{Hlawatsch1} F. Hlawatsch. Time-frequency analysis and synthesis of linear signal spaces:
time-frequency lters, signal detection and estimation, and Range-Doppler
estimation. Springer Science and Business Media (original edition: Kluwer 1998) 94 2013.

\bibitem{Hlawatsch2} F. Hlawatsch, W. Kozek. The Wigner distribution of a linear signal space.
IEEE Trans. Signal Proc. 41 (3) (1993) 1248-1258.

\bibitem{Jahn} J. Jahn. Introduction to the theory of nonlinear optimization. Springer, 2nd edition (1996).

\bibitem{Janssen} A.J.E.M. Janssen. Positivity of weighted Wigner distributions. SIAM J. Math. Anal. 12 (1981) 1-58.

\bibitem{Jeong} J. Jeong and W.J. Williams. Time-varying filtering and signal synthesis, in {\it Time-frequency Signal Analysis-Methods and Applications}. B. Boashash. ed.. Longman and Cheshire, Melbourne. Australia (1991).

\bibitem{Jia}Z. Jia, G.W. Stewart: An analysis of the Rayleigh-Ritz method for approximating eigenspaces. Mathematics of computation 70 (2000) 637-647.

\bibitem{Jost} J. Jost, X. Li-Jost. Calculus of variations. Cambridge Sudies in Advanced Mathematics (1998).

\bibitem{Koehler}F. Koehler: Estimates for the eigenvalues of infinite matrices. Pacific Journal of Mathematics 7 (1957) 1391-1404.

\bibitem{Landau} H. J. Landau. On Szegös eigenvalue distribution theorem and non-
Hermitian kernels. J. d'Analyse Math. 28 (1975) 335-357.

\bibitem{Leray} J. Leray: Lagrangian analysis and quantum mechanics. A mathematical structure related
to asymptotic expansions and the Maslov index. Translated from the French by Carolyn
Schroeder. MIT Press, Cambridge, Mass., 1981.

\bibitem{Lieb} E.H. Lieb, Y. Ostrover. Localization of multi-dimensional Wigner distributions. J. Math. Phys. 51 (2010) 102101.

\bibitem{Lions}P.L. Lions and T. Paul. Sur les mesures de Wigner. Rev. Mat. Iberoamer. 9 (1993) 553--618.

\bibitem{Loughlin} P. Loughlin, L. Cohen. Approximate wavefunction from approximate non-representable Wigner distributions. J. Mod. Optics 55 (2008) 3379-3387.

\bibitem{Loughlin1} P. Loughlin, L. Cohen. A Wigner approximation method for wave propagation. J. Acoust. Soc. Amer. 118 (2005) 1268-1271.

\bibitem{Loughlin2} P. Loughlin, L. Cohen. Local properties of dispersive pulses. J. Mod. Optics 49 (2002) 2645-2655.

\bibitem{Moyal} J. Moyal. Quantum mechanics as a statistical theory. Proc. Camb. Phil. Soc. 45 (1949) 99.

\bibitem{Narconnell}F. J. Narcowich and R. F. O'Connell. Necessary and sufficient conditions for a phase-space function to be a Wigner distribution. Phys. Rev. A 34 (1986) 1--6.

\bibitem{Narcowich} F.J. Narcowich. Conditions for the convolution of two Wigner distributions to be itself a Wigner distribution. J. Math. Phys. 29 ((1988) 2036-2041.

\bibitem{Ramanathan} J. Ramanathan, P. Topiwala. Time-frequency Localization via the Weyl correspondence. SIAM J. Math. Anal. 24 (1993) 1378-1393.

\bibitem{Reed} M. Reed, B. Simon. Methods in modern mathematical Physics. I Functional analysis. Elsevier (1980)

\bibitem{Shale}D. Shale. Linear symmetries of free boson fields. Trans. Amer. Math. Soc. 103 (1962) 149-167.

\bibitem{Slepian1} D. Slepian. Some comments on Fourier analysis, uncertainty and modeling. SIAM Rev. 25 (1983) 379-393.

\bibitem{Slepian2} D. Slepian, H.O. Pollak. Prolate spheroidal wave functions, Fourier analysis and uncertainty. I. Bell System Tech. J. 40 (1961) 43-63.

\bibitem{Tao} T. Tao. Topics in random matrix theory. Graduate studies in Mathematics, Vol. 132, AMS (2010).

\bibitem{Thangavelu} S. Thangavelu. Lectures on Hermite and Laguerre expansions. Mathematical notes 42. Princeton University Press, Princeton, NJ (1993).

\bibitem{Trefethen} L. N. Trefethen, M. Embree. Spectra and Pseudospectra: The Behavior of
Nonnormal Matrices and Operators. Princeton University Press, 2005.

\bibitem{Xiao} J. Xiao, P. Flandrin. Multitaper time-frequency reassignment for nonstation-
ary spectrum estimation and chirp enhancement. IEEE Trans. Signal Proc.
55 (6) (2007) 2851-2860.

\bibitem{Weil}A. Weil. Sur certains groupes d'opérateurs unitaires. Acta Math. 111 (1964) 143-211.

\bibitem{Weyl}H. Weyl. Das asymptotische Verteilungsgesetz der Eigenwerte linearer partieller Differentialgleichungen (mit einer Anwendung auf die Theorie der Hohlraumstrahlung). Mat. Ann. 71 (1912) 441-479.

\bibitem{Wigner} E. Wigner. On the quantum correction for thermodynamic equilibrium. Phys. Rev. 40 (1932) 749-759.

\bibitem{Williamson} J. Williamson: On the algebraic problem concerning the normal forms of linear dynamical systems. Amer. J. Math. 58 (1936) 141-163.

\bibitem{Wong} M.W. Wong. Weyl transforms. Springer-Verlag (1998).





\end{thebibliography}
\end{document}